\documentclass[11pt]{amsart}
\usepackage{amssymb,amsthm, times, multicol}
\usepackage{delarray,verbatim}
\usepackage[hmargin=1in,height=8in]{geometry}
\usepackage{xcolor}
\usepackage[colorlinks,linkcolor=blue,citecolor=magenta,anchorcolor=magenta]{hyperref}
\usepackage{mathtools}
\mathtoolsset{showonlyrefs=true}

\newcommand {\ud}{{\rm d}}

\newcommand {\rbra}[1]{\left(#1\right)}

\newcommand {\norm}[1]{\|#1\|}
\newcommand {\absol}[1]{|#1|}

\numberwithin{equation}{section}

\newtheorem{theorem}{Theorem}[section]
\newtheorem{proposition}{Proposition}[section]
\newtheorem{corollary}{Corollary}[section]
\newtheorem{remark}{Remark}[section]
\newtheorem{lemma}{Lemma}[section]

\newtheorem{assump}{Assumption}[section]

\numberwithin{remark}{section} \numberwithin{proposition}{section}
\numberwithin{corollary}{section}

\newcommand {\bR}{\mathbb{R}}
\newcommand {\R}{\mathbb{R}}

\newcommand {\cF}{\mathcal{F}}

\newcommand {\bN}{\mathbb{N}}
\newcommand {\bP}{\mathbb{P}}

\newcommand {\bE}{\mathbb{E}}
\newcommand {\E}{\mathbb{E}}
\newcommand {\cE}{\mathcal{E}}

\newcommand{\diff}{{\rm d}}

\newcommand{\lev}{L\'{e}vy }

\newcommand{\red}{\textcolor[rgb]{1.00,0.00,0.00}}

\newcommand{\bW}{\mathbb{W}}
\newcommand{\bZ}{\mathbb{Z}}
\newcommand{\cL}{\mathcal{L}}
\newcommand{\cB}{\mathcal{B}}
\newcommand{\cD}{\mathcal{D}}

\allowdisplaybreaks[1]

\title{On the bail-out dividend problem for spectrally negative Markov additive models}

\thanks{K. Noba is supported by JSPS KAKENHI grant no. JP18J12680. X. Yu is supported by Hong Kong Early Career Scheme under no. 25302116 and by Hong Kong Polytechnic University internal fund under no. P0031417.}

\author[K. Noba]{Kei Noba}
\address[K. Noba]{Graduate School of Engineering Science, Osaka University, Japan.}
\email{knoba@sigmath.es.osaka-u.ac.jp}
\author[J. L. P\'erez]{Jos\'e-Luis P\'erez}
\address[J. L. P\'erez]{Department of Probability and Statistics, Centro de Investigaci\'{o}n en Matem\'{a}ticas A.C. Calle Jalisco s/n. C.P. 36240, Guanajuato, Mexico}
\email{jluis.garmendia@cimat.mx}
\author[X. Yu]{Xiang Yu}
\address[X. Yu]{Department of Applied Mathematics, The Hong Kong Polytechnic University, Hung Hom, Kowloon, Hong Kong.}
\email{xiang.yu@polyu.edu.hk}

\begin{document}
\maketitle

\begin{abstract}
This paper studies the bail-out optimal dividend problem with regime switching under the constraint that the cumulative dividend strategy is absolutely continuous. We confirm the optimality of the regime-modulated refraction-reflection strategy when the underlying risk model follows a general spectrally negative Markov additive process. To verify the conjecture of a barrier type optimal control, we first introduce and study an auxiliary problem with the final payoff at an exponential terminal time and characterize the optimal threshold explicitly using fluctuation identities of the refracted-reflected L\'evy process. Second, we transform the problem with regime-switching into an equivalent local optimization problem with a final payoff up to the first regime switching time. The refraction-reflection strategy with regime-modulated thresholds can be shown as optimal by using results in the first step and some fixed point arguments for auxiliary recursive iterations.
\end{abstract}
\ \\
\noindent \small{\textbf{Keywords:}
Refracted-reflected spectrally negative L\'evy process, capital injection, absolutely continuous constraint, regime-switching, fixed point argument.
\ \\ 


\section{Introduction}
The bail-out version of de Finetti's optimal dividend problem has attracted a lot of research interests from the community of corporate finance and insurance. This optimal control problem is to maximize the expected net present value (NPV) of dividends with infinite time horizon while the shareholders are also required to inject capital to prevent the company from bankruptcy. A spectrally negative L\'evy process, namely a L\'evy process with only downward jumps, is usually used to describe the underlying surplus process for an insurance company that diffuses because of the premiums and jumps downside by claim payments. Avram et al. showed in \cite{APP2007} that it is optimal to inject capital by reflecting from below at zero and to pay dividends from above at a suitably chosen threshold. However, the typical admissible set of singular dividend controls is generally difficult to implement. Many practical constraints have been proposed to guarantee the optimal dividend policy has better structures. An example of such constraints requires the cumulative dividend payment to be absolutely continuous with respect to the Lebesgue measure while its density is bounded by a constant. Under this control constraint, the problem has been solved by \cite{YP2017} for spectrally positive \lev models and very recently by \cite{PerYamYu2018} for the spectrally negative L\'evy case. The optimal control under the constraint fits the type of a \textit{refraction-reflection strategy} that reflects the surplus from below at zero in the classical sense and decreases the drift of the surplus process at a suitably chosen threshold. In particular, for the spectrally negative case, the resulting controlled surplus process turns to be a refracted-reflected L\'{e}vy process introduced in \cite{PerYam2018}. Fluctuation identities of the refracted-reflected L\'{e}vy process play the core role in the verification of optimality with an explicitly chosen refraction barrier in \cite{PerYamYu2018}.

In the present paper, we are interested in the bail-out dividend problem under the same constraint as in \cite{PerYamYu2018}, but in a more general framework in which the underlying risk is characterized by a spectrally negative Markov additive process. This process can be seen as a family of L\'evy processes switching via an independent Markov chain. In addition, a negative jump is introduced each time there is a change in the current regime. This jump is independent of the family of L\'evy processes and the Markov chain, and can be understood as the cost for the insurance company to adapt to the new regime. The regime-switching model is commonly used to capture the changes of market behavior caused by macroeconomic transitions or macroscopic readjustment. On the other hand, the continuous time Markov chain is often used to approximate some stochastic factors which affect the underlying state processes. Comparing with stochastic drift or volatility models, the regime-switching model is advantageous for its tractability and more explicit structures especially in stochastic control problems. Many empirical justifications of market regimes have been conducted in various models, see a short list among \cite{AngBeK02b}, \cite{AngTim}, \cite{Gray}, \cite{Ham}, \cite{Pelle}, and \cite{So}. The generalization of control and optimization problems from single market models to models with regime-switching deserves technical treatment, which has become a vibrant research topic during the past decades. Some recent work motivated by different financial applications can be found in  \cite{ZhouYin}, \cite{ZhangZhou}, \cite{JinYin}, \cite{CapLop14a} and \cite{BoLiaoYu}.

In particular, optimal dividend problems in the context with regime-switching have been studied, however, only in the framework of diffusion models or models with Poisson jumps, see for instance \cite{Azcue}, \cite{JiaPis2012}, \cite{JinYin}, \cite{WeiWY}, \cite{ZY2016}, and \cite{ZhuJ}. Similar to the single risk model, optimal dividend strategies in these pioneer works have been shown to fit the type of barrier control as well. In the more general case, it becomes an open problem if the barrier dividend policy can still preserve its optimality.

This paper therefore aims to provide the positive answer to the optimality of the barrier dividend strategy, namely the refracted-reflected dividend and injection controls but modulated by the regime states, in general spectrally negative L\'evy models. Unlike most of the aforementioned work in the diffusion or jump-diffusion cases that rely heavily on the PDE approach, our method is purely probabilistic and is based on fluctuation identities for refracted-reflected L\'{e}vy process. On the other hand, our analysis differs from \cite{PerYamYu2018} substantially due to the complexity caused by different regimes. The verification of the optimal barrier(s) is expected to be much more involved than in \cite{PerYamYu2018} as the barrier in each regime period is clearly coupled with other regime modulated barriers through the definition of the value function. The HJB variational inequalities for the global control problem will become a system of coupled variational inequalities based on the regime states. To reduce the complexity and deal with the switch in the regimes, we borrow the idea in the literature of stochastic control to use the dynamic programming principle and localize the problem to the period up to the first regime switch, see \cite{JiaPis2012} and \cite{ZY2016} for similar optimal dividend problems.

Note that the dynamic programming principle allows us to write the problem equivalently to a bail-out dividend problem up to the first regime switching time, however, with an additional indirect utility process. Our verification of optimality can therefore be summarized in two steps:
\begin{itemize}
\item[(i)] First, we study a simplified auxiliary bail-out dividend problem with a final payoff in a single spectrally negative L\'{e}vy model up to an independent exponential time. In this part, we can successfully compute the expected NPV of dividends minus capital injection under a refraction-reflection strategy explicitly and perform the ``guess-and-verify" procedure common in the literature, see \cite{PerYamYu2018}. We construct and confirm the optimal barrier using the smooth fit principle and delicate computations of generators and slope conditions of the value function. As a byproduct, our work in the single model appears to be the first one that deals with a final payoff up to an independent exponential time for spectrally negative L\'{e}vy processes and itself is an important add-on to the literature that may have other future applications.

\item[(ii)] After the preparation in a single spectrally negative L\'{e}vy model, we can define the iteration operator by proving the dynamic programming principle similar to \cite{JiaPis2012} and \cite{ZY2016}. In our framework, we can show the existence of the candidate optimal barriers modulated by regime states using the result from step $\textrm{(i)}$. Then we proceed to prove that the corresponding expected NPV under the regime-modulated refracted-reflected strategy is the fixed point of the iteration operator. This completes the second step of the verification and the optimality of the barrier type control is successfully retained in the general model as conjectured. Our conclusion recovers many existing results in models with diffusion processes or jump diffusion processes.
\end{itemize}

The rest of the paper is organized as follows. Section \ref{sec2} introduces some preliminaries for a single spectrally negative L\'evy process. Section \ref{sec3} formulates our bail-out dividend problem with regime switching in the spectrally negative Markov additive model. The optimality of regime-modulated refraction-reflection strategies is presented as the main result of this paper. In Section \ref{sec4}, we introduce an auxiliary bail-out optimal dividend problem with absolutely continuous constraint and the final payoff at an independent exponential terminal time. The optimality of a refraction-reflection strategy is confirmed and the verification of the explicit optimal barrier is provided. In Section \ref{sec5}, an auxiliary iteration operator is defined and verification of the optimal regime-modulated thresholds is completed by using the results in Section \ref{sec4} and some fixed point arguments. The proof of the dynamic programming principle in our framework is given in Appendix \ref{secA}.

\section{Preliminaries on spectrally negative L\'evy process}\label{sec2}
Before we introduce the bail-out optimal dividend problem in spectrally negative L\'evy models with Markov regime-switching, let us first recall some preliminary results on a single spectrally negative L\'evy process and related fluctuation identities.

\subsection{Spectrally negative L\'evy process}\label{SNP_intro}
Let us consider a spectrally negative L\'evy process $X = (X(t) ; t \geq 0)$  defined on some filtered probability space
$(\Omega , \cF , {\bf F} , \bP )$, where ${\bf F} := \{ \cF(t), t \geq 0\}$ denotes the right-continuous complete filtration
generated by the process $X$.
We will write by $\bP_{x}$ the law of the process conditioned on the event $\{ X_0=x  \}$ and $\bE_x$ as the associated expectation operator.

Let $\psi_X : [0, \infty) \rightarrow \bR$ be the Laplace exponent of the L\'evy process $X$, i.e. $\bE_0 \left[e^{\theta X(t)}\right] =: e^{\psi_{X}(\theta )t}$,  $t, \theta \geq 0$. We will assume throughout the paper that $\psi_X$ is given by the L\'evy--Khintchine formula
\begin{align*}
\psi_{X} (\theta) := \gamma \theta +\frac{\sigma^2}{2} \theta^2 + \int_{(-\infty , 0)} (e^{\theta z}-1 - \theta z1_{\{ z>-1\}})\Pi (dz),
~~~\theta \geq 0,
\end{align*}
where $\gamma \in \bR$, $\sigma \geq 0$ and $\Pi$ is the L\'evy measure of $X$ on $(-\infty,0)$ that satisfies $\int_{(-\infty,0)} (1 \land x^2) \Pi (dx) < \infty$. It is well-known that $X$ has paths of bounded variation if and only if $\sigma=0$ and $\int_{(-1,0)} |z|\Pi(\mathrm{d}z)$ is finite.  In this case, its Laplace exponent is given simply by
\begin{equation}\label{psi_bounded_var}
\psi_X(\theta) = c \theta+\int_{(-\infty,0)}\big( {\rm e}^{\theta z}-1\big)\Pi(\ud z), \quad \theta \geq 0, 
\end{equation}
where $c:=\gamma-\int_{(-1,0)} z\Pi(\mathrm{d}z)$. Note that necessarily $c>0$ because we have ruled out the case that $X$ has monotone paths.

\subsection{Scale functions}
In order to solve the stochastic control problems in later sections, we need to introduce the so-called scale functions using fluctuation identities for spectrally negative L\'evy processes. For $q\geq 0$, let $W^{(q)}$ denote the $q$-scale function of the process $X(t)$ and $\bW^{(q)}$ be the $q$-scale function of the process $Y(t)=X(t)-\delta t$. These are the mappings from $\bR$ to $[0, \infty)$ that take value zero on the negative half-line, while on the positive half-line
they are strictly increasing functions that are defined by their Laplace transforms:
\begin{align}
\int_0^\infty e^{-\theta x}W^{(q)}(x)dx:=\frac{1}{\psi_X (\theta) -q},\qquad\text{$\theta > \Phi (q)$},\label{111a}\\
\int_0^\infty e^{-\theta x}\bW^{(q)}(x)dx:=\frac{1}{\psi_Y (\theta) -q},\qquad\text{$\theta > \varphi (q)$}\label{111b}.
\end{align}
Here $\psi_Y(\theta) = \psi_X (\theta) - \delta \theta$, $\theta\geq0$, is the Laplace exponent of the process $Y$, and
\begin{align*}
\Phi (q) := \sup\{ \lambda \geq 0 : \psi_X(\lambda) = q\}\ \quad\text{and} \qquad \varphi (q) := \sup \{ \lambda \geq 0 :\psi_Y(\lambda)=q \}.
\end{align*}

We also define, for $x \in \R$,
\begin{align*}
\overline{W}^{(q)}(x) &:=  \int_0^x W^{(q)}(y) \diff y, \quad
Z^{(q)}(x) := 1 + q \overline{W}^{(q)}(x),  \\
\overline{Z}^{(q)}(x) &:= \int_0^x Z^{(q)} (z) \diff z = x + q \int_0^x \int_0^z W^{(q)} (w) \diff w \diff z.
\end{align*}
Noting that $W^{(q)}(x) = 0$ for $-\infty < x < 0$, we have
\begin{align}
\overline{W}^{(q)}(x) = 0, \quad Z^{(q)}(x) = 1  \quad \textrm{and} \quad \overline{Z}^{(q)}(x) = x, \quad x \leq 0.  \label{z_below_zero}
\end{align}
Analogously, we define $\overline{\mathbb{W}}^{(q)}$, $\mathbb{Z}^{(q)}$ and $\overline{\mathbb{Z}}^{(q)}$ for the process $Y$.

From the results in \cite{KL}, for $q, x\geq 0$, we have the following identities between the scale functions $W^{(q)}$ and $\mathbb{W}^{(q)}$,
\begin{align}
\delta \int_0^x \bW^{(q)}(x-y)W^{(q)}(y) dy&= \overline{\bW}^{(q)}(x)-\overline{W}^{(q)}(x), \label{111}\\
\delta \int_0^x \bW^{(q)}(y)W^{(q)\prime}(x-y) dy&= (1-\delta W^{(\alpha)}(0))\bW^{(q)}(x)-W^{(q)}(x).\label{111_der}
\end{align}

\begin{remark} \label{remark_scale_function_properties}
	\begin{enumerate}
		\item $W^{(q)}$ and $\mathbb{W}^{(q)}$ are differentiable a.e. In particular, if $X$ is of unbounded variation or the \lev measure is atomless, it is known that $W^{(q)}$ and $\mathbb{W}^{(q)}$ are $C^1(\R \backslash \{0\})$; see, e.g.,\ \cite[Theorem 3]{Chan2011}.
		\item As $x \downarrow 0$, by Lemma 3.1 of \cite{KKR}, we have
		\begin{align*}
		\begin{split}
		W^{(q)} (0) &= \left\{ \begin{array}{ll} 0, & \textrm{if $X$ is of unbounded
			variation,} \\ c^{-1}, & \textrm{if $X$ is of bounded variation,}
		\end{array} \right.  \\ \mathbb{W}^{(q)} (0) &= \left\{ \begin{array}{ll} 0, & \textrm{if $Y$ is of unbounded
			variation,} \\ (c-\delta)^{-1}, & \textrm{if $Y$ is of bounded variation.}
		\end{array} \right.
		\end{split}
		\end{align*}
		\item As in Lemma 3.3 of \cite{KKR}, we also have
        \begin{align} \label{W_zero_derivative}
		\begin{split}
        W^{(q)\prime}_+ (0) &:= \lim_{x \downarrow 0}W^{(q)'}_+ (x) =
		\left\{ \begin{array}{ll}  \frac 2{\sigma^2}, & \textrm{if }\sigma > 0, \\
		\infty, & \textrm{if }\sigma = 0 \; \textrm{and} \; \Pi(0, \infty) = \infty, \\
		\frac {q + \Pi(0, \infty)} {c_X^2}, &  \textrm{if }\sigma = 0 \; \textrm{and} \; \Pi(0, \infty) < \infty,
		\end{array} \right.\\
		\mathbb{W}^{(q) \prime}_+ (0) &:= \lim_{x \downarrow 0}\mathbb{W}^{(q)\prime}_+ (x) =
		\left\{ \begin{array}{ll}  \frac 2 {\sigma^2}, & \textrm{if }\sigma > 0, \\
		\infty, & \textrm{if }\sigma = 0 \; \textrm{and} \; \Pi(0,\infty) = \infty, \\
		\frac {q + \Pi(0, \infty)} {c_Y^2}, &  \textrm{if }\sigma = 0 \; \textrm{and} \; \Pi(0,\infty) < \infty.
		\end{array} \right.
		\end{split}
		\end{align}
		
	\end{enumerate}
\end{remark}

\section{Bail-out optimal dividend problem with regime switching}\label{sec3}
We now introduce and formulate our dividend problem in a general setting with a spectrally negative Markov additive model and present our main result that states the optimality of the barrier strategy.

\subsection{Spectrally negative Markov additive processes}
Let us consider a bivariate process $((X(t),H(t)) ; t \geq 0)$, where the component $H$ is a continuous-time Markov chain with finite state space $E$ and the generator matrix $Q={(q_{ij})}_{i, j\in E}$.
When the Markov chain $H$ is in the state $i$, the process $X$ behaves as a spectrally negative L\'evy process $X^i$. In addition, when the process $H$ changes to a state $j\not=i$, the process $X$ jumps according to a non-positive random variable $J_{ij}$ with $i,j\in E$. The components $(X^i)_{i\in E}$, $H$, and $(J_{ij})_{i,j\in E}$ are assumed to be independent and are defined on some filtered probability space
$(\Omega , \cF , {\bf F} , \bP )$ where ${\bf F} := \{ \cF(t), t \geq 0\}$ denotes the right-continuous complete filtration jointly
generated by the processes $(X,H)$ and the family of random variables $(J_{ij})_{i,j\in E}$.
We will denote by $\bP_{(x, i)}$ the law of the process conditioned on the event $\{ X_0=x , H_0 =i \}$.
\par
We will assume throughout this paper that for each $i\in E$ the Laplace exponent of the L\'evy process $X^i$,
$\psi^i : [0, \infty) \rightarrow \bR$, i.e.
\begin{align*}
	\bE_0 \left[e^{\theta X^i(t)}\right] =: e^{\psi_{X^i}(\theta)t}, ~~~t, ~ \theta \geq 0,
\end{align*}
is given by the L\'evy--Khintchine formula
\begin{align*}
\psi_{X^i} (\theta) := \gamma(i) \theta +\frac{\sigma(i)^2}{2} \theta^2 + \int_{(-\infty , 0)} (e^{\theta z}-1 - \theta z1_{\{ -1<z<0\}})\Pi (i,dz),
~~~\theta \geq 0,
\end{align*}
where $\gamma_i \in \bR$, $\sigma_i \geq 0$ and $\Pi_i$ is a measure on $(-\infty,0)$ called the L\'evy measure of $X^i$ that satisfies $\int_{(-\infty,0)} (1 \land x^2) \Pi(i,dx) < \infty$. As in Section \ref{SNP_intro}, if $X^i$ is of bounded variation, its Laplace exponent is given by
$\psi_{X^i}(\theta) = c (i) \theta+\int_{(-\infty,0)}\big( {\rm e}^{\theta z}-1\big)\Pi(i, \ud z)$, $\theta \geq 0$,
where $c(i):=\gamma (i) -\int_{(-1,0)} z\Pi(i, \mathrm{d}z)$.

\subsection{Bail-out optimal dividend problem with absolutely continuous strategies and Markov switching regimes}
We consider a bail-out de Finetti's dividend problem, namely the shareholders need to inject capital to prevent the company from going bankrupt.  In particular, a strategy  is a pair $\pi := \left( L^{\pi}(t), R^{\pi}(t); t \geq 0 \right)$ of non-decreasing, right-continuous, and adapted processes (with respect to the filtration generated by $X$ and $H$) starting at zero where $L^{\pi}$ is the cumulative amount of dividends and $R^{\pi}$ is that of injected capital. With $U^\pi(0-) := x$ and $U^\pi(t) := X(t) - L^\pi(t) + R^\pi(t)$, $t \geq 0$, it is required that $U^\pi(t) \geq 0$ a.s.\ uniformly in $t$. For a given function $\delta: E\mapsto \R_+$, we require that $L^\pi$ is absolutely continuous with respect to the Lebesgue measure of the form $L^\pi(t) = \int_0^t \ell^\pi(s) \diff s$, $t \geq 0$, with $\ell^\pi$ restricted to take values in $[0,\delta(H(t))]$ uniformly in time.

About the capital injection $R^\pi$, it is assumed that
\begin{align}\label{adm_ci}
\bE_{(x, i)} \left[ \int_{[0, \infty)} e^{-\int_0^t r(H(s)) ds} dR^\pi(t)\right]< \infty, \forall\ x\geq 0, i\in E.
\end{align}

Similar to \cite{ZY2016}, it is considered in the present paper that $\beta>1$,  which represents the constant cost per unit of injected capital in all regimes. We want to maximize
\begin{align*}
V_{\pi} (x,i) := \mathbb{E}_{(x,i)} \left( \int_0^\infty e^{-\int_0^tr(H(s))ds} \ell^\pi(s)  \diff t -  \int_{[0, \infty)} e^{-\int_0^tr(H(s))ds} \beta \diff R^{\pi}(t)\right), \quad x \geq 0,
\end{align*}
where $r:E\mapsto \R_+$ represents the Markov-modulated rate of discounting.
Hence our aim is to find the value function of the problem, i.e.,
\begin{equation}\label{vf_rs}
V(x,i):=\sup_{\pi \in \mathcal{A}}V_{\pi}(x,i), \quad x \geq 0,
\end{equation}
where $\mathcal{A}$ is the set of all admissible strategies that satisfy constraints described as above.

For the problem with regime-switching in this section, the following assumptions are mandated.
\begin{assump}\label{AAA1}
	We assume that $\E\left[X^i_1\right]=\psi_{X^i}'(0+)>-\infty$ for all $i\in E$.
\end{assump}
\begin{assump}\label{AAA2}
For all $i, j\in E$ with $i\neq j$, we assume that $\max_{i,j\in E}\E[-J_{ij}]<\infty$.
\end{assump}


On the other hand, to avoid that the process $Y^i:=\{ Y^i(t) = X^i(t) -\delta(i) t: t\geq 0\}$ has monotone paths for $i\in E$, we make the following assumption.
\begin{assump}\label{AAA4}
	If the process $X^i$ has bounded variation paths, we require that $c(i)>\delta(i)$ for $i\in E$.
\end{assump}	

In the next result, we claim that the dynamic programming principle for the value function of our control problem holds valid, which will play a key role in finding the fixed point of some recursive iterations later on. The proof will be reported in Appendix \ref{secA}.
\begin{proposition}\label{Prop101}
For $x\in\bR$ and $i\in E$, we have
\begin{align}
{V(x, i)} = &
\sup_{\pi \in \Pi}
\bE_{(x, i)} \bigg{[}\int_0^{{\zeta} } e^{- {\Lambda(t) }}l^\pi(t) dt
-\int_{[0, {\zeta}]} {\beta}e^{-{\Lambda(t)}}dR^\pi(t)
\notag\\&+
 e^{-\Lambda (\zeta)}V(U^\pi (\zeta), H(\zeta))\bigg{]}, \label{4}
\end{align}
where $\zeta$ denotes the epoch of the first regime switch and $\Lambda(t) :=\int_0^tr(H(s))ds$.
\end{proposition}

\subsection{Markov-modulated refraction-reflection strategies}
For the optimal control, we will consider the Markov-modulated refraction-reflection strategy, say $\pi^{b} = ((L^{0,b}(t), R^{0,b}(t)), t \geq 0)$, with a suitable refraction level $b=(b(i),i\in E)$. Namely, dividends are paid at the maximal rate $\delta(H(t))$ whenever the surplus process is above $b(H(t))$ while it is pushed upward by capital injection whenever it attempts to downcross zero. The resulting surplus process becomes the Markov-modulated refracted-reflected \lev process given by
$U^{0,b}(t) := X(t) - L^{0,b}(t) + R^{0,b}(t)$, $t\geq 0$.
We can explicitly describe the cumulative dividend control modulated by regime states as
\begin{align*}
L^{0,b} (t)= \int_0^t \delta(H(s)) 1_{\{ U^{0,b}(t) > b(H(s)) \}} \diff s\ =:\int_0^t l^{0,b}(s)\diff s.
\end{align*}
Let $\cE$ be a set of functions from $E$ to $[0, \infty)$. The next theorem is the main result of this paper and its proof will be provided by an iterative construction of the value function $V$ in Section \ref{sec5}.

\begin{remark}
The Markov-modulated refraction-reflection strategy $\pi^b$ is indeed admissible. To wit, if we set $T_0=0$, and for each $n\geq 1$ we denote by $T_n$ the $n$-th jump time of $H$, then by the strong Markov property
\begin{align*}
\bE_{(x, i)} \left[ \int_{[0, \infty)} e^{-\int_0^t r(H(s)) ds} dR^{\pi^b}(t)\right]
=&
\sum_{n\geq 1}
\bE_{(x, i)} \left[ \int_{[T_{n-1}, T_n)} e^{-\int_0^t r(H(s)) ds} dR^{\pi^b}(t)\right] \\
=&M_b+
\sum_{n\geq 1}
M_e^{n}(M_J + M_b )<\infty, 
\end{align*}
where
\begin{align*}
M_b:=\max_{i\in E}\bE_{(0, i)} \left[ \int_{[0, T_1)} e^{-r(i) t} dR^{\pi^b}(t)\right],\ 
M_J:=\max_{i,j\in E}\E[-J_{ij}],\ 
\text{and}\ \ 
M_e:= \max_{i \in E} \E_{(0, i)} \left[   e^{-r(i)T_1}\right]. 
\end{align*}	
Therefore by Assumption \ref{AAA1}, Assumption \ref{AAA2} and the admissibility of the refraction-reflection strategies, we conclude that \eqref{adm_ci} holds.
\end{remark}

\begin{theorem}\label{Thm201}
Under Assumptions \ref{AAA1}, \ref{AAA2} 
and \ref{AAA4}, there exists a function $b^\ast \in \cE$, such that the Markov-modulated refracted-reflected strategy $\pi^{b^*}$ is optimal and the value function of the problem \eqref{vf_rs} is given by
\begin{align*}
V(x, i) = V_{\pi^{b^\ast}} (x, i),\qquad\textrm{for}\ x \in [0, \infty).
\end{align*}
\end{theorem}

\section{Optimal strategies for an auxiliary bail-out dividend problems until an exponential terminal time}\label{sec4}
In order to prove Theorem \ref{Thm201} above,
we will first introduce and study an auxiliary bail-out optimal dividend problem with a final payoff at an independent exponential terminal time in a single spectrally negative L\'evy model. 

\subsection{Bail-out optimal dividend problem with absolutely continuous strategies and an exponential terminal time}
In a model with a single spectrally negative L\'evy process $X(t)$,  let $\pi = \{ (L^\pi(t), R^\pi(t)) : t \geq 0 \}$ be non-decreasing, right-continuous and adapted processes with respect to $\bf F$ starting at zero
where $L^\pi$ is the cumulative amount of dividends and $R^\pi$ is that of injected capital.
With $U^\pi (0-):=x$ and
\begin{align}
U^\pi(t):= X(t)-L^\pi(t)+R^\pi(t), ~~~t\geq 0,
\end{align}
it is required that $U^\pi(t) \geq 0$ a.s. uniformly in $t$.
Let us consider a constant $\delta >0$, then we require that $L^\pi$ is absolutely continuous with respect to
the Lebesgue measure of the form $L^\pi(t) = \int_0^t l^\pi(s) ds$, $t\geq 0$, with $l^\pi$ restricted to take values in $[0, \delta]$
uniformly in time. About the capital injection $R^\pi$, it is assumed that for $q>0$ and $x\geq0$,
\begin{align}
\bE_{x} \left[ \int_{[0, \zeta)} e^{-qt} dR^\pi(t)\right]< \infty , \label{9}
\end{align}
where $\zeta$ is an independent exponential random variable of parameter $r>0$.

For $\beta> 1$ as the cost per unit of injected capital and $q>0$ as the discount factor, the expected net present value (NPV) of the dividend payments minus the cost of capital injection associated with a strategy $\pi$ with initial capital $x\geq 0$ is defined by
\begin{align}
v_\pi (x) := &\bE_{x} \left[\int_0^\zeta e^{-qt}\diff L^\pi(t)  -\beta\int_{[0, \zeta)} e^{-qt}dR^\pi(t)
+e^{-q\zeta} w(U^\pi(\zeta))\right]\notag\\
=&r \int_0^\infty e^{-rs}\bE_{x} \left[\int_0^s e^{-qt}\diff L^\pi(t)  -\beta\int_{[0, s)} e^{-qt}dR^\pi(t)
+e^{-qs} w(U^\pi(s))\right]ds\notag\\
=&\bE_{x} \left[\int_0^\infty e^{-\alpha t}\diff L^\pi(t) -\beta\int_{[0, \infty)} e^{-\alpha t}dR^\pi(t)
+r\int_0^\infty e^{-\alpha t} w(U^\pi(t))dt\right],\label{value_func}
\end{align}
where $\alpha= q + r$ and $w\in C[0, \infty)$ is a function representing the payoff at an exponential time. It is worth noting that $\bP_x$ is the law of $X$ and $\zeta$ is an independent random variable such that $X$ does not jump at $\zeta$, therefore we have $U^\pi (\zeta -)= U^\pi(\zeta)$, $\bP_x$-a.s.

The value function for the stochastic control problem is then formulated as
\begin{align}
v_{\pi^\ast} (x) := \sup_{\pi ^\in \Pi} v_\pi (x),\qquad\text{$x\in\R$,}\label{vf_def}
\end{align}
where $\Pi$ denotes the set of all admissible strategies that satisfy the constraints described previously. We also aim to obtain the optimal dividend strategy $\pi^*$ which achieves the value function in \eqref{vf_def}.

Throughout this section, we will make the next standard assumption on the L\'evy process $X$.
\begin{assump}\label{AA1}
	We assume that $\E\left[X_1\right]=\psi'(0+)>-\infty$.
\end{assump}
\par 
A common assumption in the literature (see \cite{KLP,PerYamYu2018}) is also required to avoid that the process $Y:=\{ Y(t) = X(t) -\delta t: t\geq 0\}$ has monotone paths.
\begin{assump}\label{AA2}
	For the case that the process $X$ has bounded variation paths, it is assumed that $c>\delta$.
\end{assump}
For the payoff function $w(\cdot):[0, \infty)\to\bR$, we also require the following conditions.
\begin{assump}\label{AA3}
	It is assumed that $w$ is continuous on $[0,\infty)$ and concave with $w'_+(0+)\leq \beta$ and $w'_+(\infty):=\lim_{x\to\infty}w'_+(x)\in[0,1]$, where we use $w'_+(x)$ to represent the right derivative of a concave function $w(x)$.
\end{assump}
Starting from this point onwards, we will distinguish the two notations $f'(x)$ and $f'_+(x)$ for a given function $f(x)$, the former one denotes its standard derivative and the latter one denotes its right derivative.  

\begin{remark}
If we assume that the final payoff function $w(x)\equiv 0$ in the objective function \eqref{value_func}, the bail-out optimal dividend problem \eqref{vf_def} is essentially the same as the one studied in \cite{PerYamYu2018} after replacing the parameter $q$ in \cite{PerYamYu2018} by $q+r$ in the present paper. This motivated us to conjecture that the mathematical argument and the optimality of the barrier strategy in \cite{PerYamYu2018} may also hold in this paper with a final payoff function after some technical modifications. 

Let us first recall the ``guess-and-verify'' procedure applied in \cite{PerYamYu2018}, which crucially relies on the fluctuation identities for refracted-reflected L\'{e}vy processes developed in \cite{PerYam2018}: (i) one first selects a candidate refraction-reflection strategy and its associated barrier level via the smooth fit principle such that the NPV becomes continuously (resp. twice continuously) differentiable at the threshold for the case of bounded (resp. unbounded) variation; (ii) the optimality of the selected strategy is confirmed by verifying the variational inequalities that require the computation of the generators and certain slope conditions of the value function. In principle, as long as $w(x)\equiv 0$, one can simply apply some probabilistic killing arguments as in \cite{JeI} and modify some proofs in \cite{PerYam2018} and \cite{PerYamYu2018} to conclude the optimality of a barrier strategy in problem \eqref{vf_def}. Actually, in view of the exponential distribution of $\zeta$, all results in \cite{PerYamYu2018} hold true for problem \eqref{vf_def} in a trivial way with the revised parameter $q+r$.

As opposed to \cite{PerYamYu2018}, the presence of the payoff function $w(x)$ in \eqref{value_func} renders the problem more difficult and the probabilistic killing arguments in \cite{JeI} can not be applied directly. We aim to follow the aforementioned ``guess-and-verify'' procedure. However, most computations differ substantially from \cite{PerYamYu2018} and the verification part in step (ii) becomes more involved as it relies on some new fluctuation identities and the concavity of the payoff function $w(x)$ as in Assumption \ref{AA3}, see for instance the proof of Lemma \ref{lem4-1} and Lemma \ref{lemma_slope}. To make the presentation complete and self-contained, we will formally perform a ``guess-and-verify'' procedure step by step to incorporate the payoff function $w(x)$ in the subsequent subsections.

\end{remark}

\subsection{Refraction--Reflection strategies for spectrally negative L\'evy processes}\label{Sec103}
We proceed to verify the optimality of a refraction-reflection strategy $\pi^b =\{  (L^{(0,b)}(t), R^{(0,b)}(t))  :  t \geq 0\}$ with an appropriate refraction threshold $b\geq0$. Under this strategy, dividends are paid at the maximal rate $\delta>0$ whenever the surplus process is above the level $b$ while the process is pushed upward by injecting capital whenever the process attempts to downcross below $0$. The resulting surplus process
\[
U^{(0,b)}(t)=X(t)-L^{(0,b)}(t)+R^{(0,b)}(t),\qquad\text{$t \geq0$,}
\]
is the so-called refracted-reflected spectrally negative L\'evy process, which was first introduced and studied in \cite{PerYam2018}.

The cumulative dividend control can be explicitly expressed as $L^{(0,b)}(t) = \int_0^t \delta 1_{\{ U^{0,b}(s) > b \}} \diff s$,
and for the case of bounded variation, we can write the candidate capital injection as
\begin{align*}
R^{(0,b)}(t) = \sum_{0 \leq s \leq t} |U^{0,b}(s-)| 1_{\{U^{0,b}(s-) < 0\}}.
\end{align*}
For a formal construction of this process, we refer the reader to \cite{PerYam2018}. For any $b \geq 0$, the admissibility of the aforementioned refraction-reflection strategy $\pi^{b}$ follows from Lemma 3.1 in \cite{PerYamYu2018}.

We denote the corresponding expected NPV by
\begin{align}
v_{b} (x) := &\bE_{x} \left[\int_0^\zeta e^{-qt}\diff L^{(0,b)}(t)  -\beta\int_{[0, \zeta)} e^{-qt}dR^{(0,b)}(t)
	+e^{-q\zeta} w(U^{(0,b)}(\zeta)) \right]\notag\\
=&\bE_{x} \left[\int_0^\infty e^{-\alpha t}\diff L^{(0,b)}(t) -\beta\int_{[0, \infty)} e^{-\alpha t}dR^{(0,b)}(t)
	+r\int_0^\infty e^{-\alpha t} w(U^{(0,b)}(t))dt\right].\label{v_pi}
	\end{align}
\begin{remark}\label{bounded_v}
	If the function $w:[0,\infty)\mapsto\R$ is bounded, then by using \eqref{9} together with \eqref{v_pi}, we have that for any $b\geq0$,
\begin{align*}
|v_{b} (x)|&\leq \frac{\delta}{\alpha}+\beta \bE_{x} \left[\int_{[0, \infty)} e^{-\alpha t}dR^{(0,b)}(t)\right]+\sup_{x\geq0}|w(x)|\frac{r}{\alpha}\\&\leq \frac{1}{\alpha}\left(\delta+r\sup_{x\geq0}|w(x)|\right)+\beta \bE_{0} \left[\int_{[0, \infty)} e^{-\alpha t}dR^{(0,b)}(t)\right]<\infty,\qquad\text{$x\geq0$,}
\end{align*}
where the second inequality follows from the fact that the mapping $x\mapsto\bE_{x} \left[\int_{[0, \infty)} e^{-\alpha t}dR^{(0,b)}(t)\right]$ is non-increasing. Therefore $v_b$ is bounded.
\end{remark}

The next result confirms the optimality of the barrier type dividend control for the auxiliary problem.
\begin{theorem}\label{Prop102}
Under Assumptions \ref{AA1}, \ref{AA2} and \ref{AA3}, there exists a constant barrier $0\leq b^*<\infty$ such that the refracted-reflected strategy at the level $b^*$, i.e. $\pi^{b^*}$, is optimal and the value function is given by $v(x)=v_{b^*}(x)=v_{\pi^*}(x)$ for all $x\geq 0$.
\end{theorem}
We now provide the construction and verification of the optimal strategy in Theorem \ref{Prop102} in the subsequent subsections.
\subsection{The computation of $v_b$}
Using the results in \cite{PerYam2018}, we have the following equivalent form of the expected NPV \eqref{v_pi} for a refracted-reflected strategy at the level $b\geq0$.
\begin{proposition}\label{enp_aux}
For $q>0$, $b\geq0$ and $x\in \R$, we have
\begin{align}
v_{b}(x)
=&-\delta \overline{\bW}^{(\alpha)}(x- b)+ \beta\rbra{\overline{Z}^{(\alpha)} (x) + \frac{\psi_X^\prime(0+)}{\alpha}}
+\beta \delta \int_b^x \bW^{(\alpha)}(x-y) Z^{(\alpha)}(y) dy \notag\\
&-\frac{\beta Z^{(\alpha)} (b) -1+\beta \alpha \int_0^\infty e^{-\varphi (\alpha) y}W^{(\alpha)}(y+b)dy}
{\alpha \varphi (\alpha)\int_0^\infty e^{-\varphi (\alpha) y}W^{(\alpha)}(y+b)dy}\notag\\
&\times\rbra{Z^{(\alpha)} (x) +\alpha \delta \int_b^x \bW^{(\alpha)}(x-y)W^{(\alpha)}(y)dy}\notag\\
&+r\int_{(0, b)}w'_+(z) \bigg{(}
\frac{ \int_b^\infty e^{-\varphi (\alpha) u}W^{(\alpha)}(u-z)du }
{ \alpha \int_b^\infty e^{-\varphi (\alpha) y} W^{(\alpha)} (y)dy}
\rbra{ Z^{(\alpha)}(x) +\alpha \delta \int_b^x \bW^{(\alpha)} (x- y ) W^{(\alpha)} (y) dy  }\notag\\
&- \rbra{\overline{W}^{(\alpha)} (x-z) + \delta \int_b^x \bW^{(\alpha)} (x - y) W^{(\alpha)}(y-z)dy}
\bigg{)} dz+\frac{rw(0)}{\alpha}\notag\\
&+r\int_{(b, \infty)}w'_+ (z) \bigg{(}\frac{e^{-\varphi (\alpha)z }  }
{\varphi (\alpha)\delta \alpha \int_b^\infty e^{-\varphi (\alpha) y} W^{(\alpha)} (y)dy}\notag \\
&\times \rbra{ Z^{(\alpha)}(x) +\alpha \delta \int_b^x \bW^{(\alpha)} (x- y ) W^{(\alpha)} (y) dy  }
-\overline{\bW}^{(\alpha)}(x- z) \bigg{)} dz.\label{enpv_aux}
\end{align}
Here we remind the reader that $w'_+(x)$ denotes the right derivative of the concave final payoff function $w(x)$ as we do not impose the differentiability of $w(x)$. 
\end{proposition}
\begin{proof}
By Lemma 3.1 in \cite{PerYamYu2018}, we have
\begin{align}
z_1(x):=&\bE_{x} \left[\int_0^\infty e^{-\alpha t}\diff L^{(0,b)}(t) -\beta\int_{[0, \infty)} e^{-\alpha t}dR^{(0,b)}(t) \right] \notag\\
=&-\delta \overline{\bW}^{(\alpha)}(x- b)+ \beta\rbra{\overline{Z}^{(\alpha)} (x) + \frac{\psi_X^\prime(0+)}{\alpha}}
+\beta \delta \int_b^x \bW^{(\alpha)}(x-y) Z^{(\alpha)}(y) dy \notag\\
&-\frac{\beta Z^{(\alpha)} (b) -1+\beta \alpha \int_0^\infty e^{-\varphi (\alpha) y}W^{(\alpha)}(y+b)dy}
{\alpha \varphi (\alpha)\int_0^\infty e^{-\varphi (\alpha) y}W^{(\alpha)}(y+b)dy}\notag\\
&\times\rbra{Z^{(\alpha)} (x) +\alpha \delta \int_b^x \bW^{(\alpha)}(x-y)W^{(\alpha)}(y)dy}. \label{113}
\end{align}
Thanks to Corollary 4.1 in \cite{PerYam2018}, we also have
\begin{align}
z_2(x):=&\bE_x\left[r\int_0^\infty e^{-\alpha t} w(U^{(0,b)}(t))dt\right]\notag\\
=&
r\int_{(0, \infty)}w(z) \Bigg(\frac{e^{-\varphi (\alpha)z } 1_{\{b<z \}}+\delta 1_{\{ 0< z< b\}}
\int_b^\infty e^{-\varphi (\alpha) u}W^{(\alpha)\prime}(u-z)du }
{\delta \alpha \int_b^\infty e^{-\varphi (\alpha) y} W^{(\alpha)} (y)dy} \notag\\
&\times \rbra{ Z^{(\alpha)}(x) +\alpha \delta \int_b^x \bW^{(\alpha)} (x- y ) W^{(\alpha)} (y) dy  }\notag\\
&-\bigg[1_{\{0<z<b \}} \rbra{W^{(\alpha)} (x-z) + \delta \int_b^x \bW^{(\alpha)} (x - y) W^{(\alpha)\prime}(y-z)dy}\notag\\
&+1_{\{ b<z<x\}}\bW^{(\alpha)}(x- z)\bigg] \Bigg)dz \notag\\
=&rw(b)\rbra{ Z^{(\alpha)}(x) +\alpha \delta \int_b^x \bW^{(\alpha)} (x- y ) W^{(\alpha)} (y) dy  }\notag\\
&\times\frac{ e^{-\varphi (\alpha)b}-\varphi (\alpha) \delta \int_b^\infty e^{-\varphi (\alpha) u}W^{(\alpha)}(u-b)du }
{\varphi (\alpha) \delta \alpha \int_b^\infty e^{-\varphi (\alpha) y} W^{(\alpha)} (y)dy} +\frac{rw(0)}{\alpha}\notag\\
&+r\int_{(0, b)}w'_+(z) \Bigg{(}
\frac{ \int_b^\infty e^{-\varphi (\alpha) u}W^{(\alpha)}(u-z)du }
{ \alpha \int_b^\infty e^{-\varphi (\alpha) y} W^{(\alpha)} (y)dy}
\rbra{ Z^{(\alpha)}(x) +\alpha \delta \int_b^x \bW^{(\alpha)} (x- y ) W^{(\alpha)} (y) dy  }\notag\\
&- \rbra{\overline{W}^{(\alpha)} (x-z) + \delta \int_b^x \bW^{(\alpha)} (x - y) W^{(\alpha)}(y-z)dy}
\Bigg{)} dz\notag\\
&+r\int_{(b, \infty)}w'_+(z) \Bigg{(}\frac{e^{-\varphi (\alpha)z }  }
{\varphi (\alpha)\delta \alpha \int_b^\infty e^{-\varphi (\alpha) y} W^{(\alpha)} (y)dy} \notag\\
&\times \rbra{ Z^{(\alpha)}(x) +\alpha \delta \int_b^x \bW^{(\alpha)} (x- y ) W^{(\alpha)} (y) dy  }
-\overline{\bW}^{(\alpha)}(x- z) \Bigg{)} dz, \label{117}
\end{align}
where the third equality follows by integration by parts together with \eqref{111}.

The fact that $\psi_X(\theta)=\psi_Y(\theta)+\delta\theta$ for $\theta\geq0$ and \eqref{111a} imply
\begin{align*}
\int_b^\infty e^{-\varphi (\alpha)u}W^{(\alpha)}(u-b)du
=&e^{-\varphi (\alpha) b}\int_0^\infty e^{-\varphi (\alpha) u}W^{(\alpha)}(u)du \notag\\
=&e^{-\varphi (\alpha) b}\frac{1}{\psi_Y(\varphi(\alpha))+\delta \varphi (\alpha)-\alpha}=\frac{e^{-\varphi (\alpha) b}}{\delta \varphi (\alpha)}.
\end{align*}
Hence, using the previous identity in \eqref{117} gives
\begin{align*}
z_2(x)=&\frac{rw(0)}{\alpha}\\
&+r\int_{(0, b)}w'_+(z) \Bigg{(}
\frac{ \int_b^\infty e^{-\varphi (\alpha) u}W^{(\alpha)}(u-z)du }
{ \alpha \int_b^\infty e^{-\varphi (\alpha) y} W^{(\alpha)} (y)dy}
\rbra{ Z^{(\alpha)}(x) +\alpha \delta \int_b^x \bW^{(\alpha)} (x- y ) W^{(\alpha)} (y) dy  }\\
&- \rbra{\overline{W}^{(\alpha)} (x-z) + \delta \int_b^x \bW^{(\alpha)} (x - y) W^{(\alpha)}(y-z)dy}
\Bigg{)} dz\\
&+r\int_{(b, \infty)}w'_+(z) \Bigg{(}\frac{e^{-\varphi (\alpha)z }  }
{\varphi (\alpha)\delta \alpha \int_b^\infty e^{-\varphi (\alpha) y} W^{(\alpha)} (y)dy} \\
&\times \rbra{ Z^{(\alpha)}(x) +\alpha \delta \int_b^x \bW^{(\alpha)} (x- y ) W^{(\alpha)} (y) dy  }
-\overline{\bW}^{(\alpha)}(x- z) \Bigg{)} dz.
\end{align*}
Finally, the fact that $v_{b}(x)=z_1(x)+z_2(x)$ yields \eqref{enpv_aux}.
\end{proof}

\subsection{Selection of the candidate threshold}
In this section, we aim to find the candidate optimal threshold $b^*$ such that the associated expected NPV $v_{b^*}$ defined in \eqref{v_pi}, is smooth at the threshold $b^{*}$. To this end, by differentiating \eqref{enpv_aux} and using dominated convergence, we first obtain that for $x \in (0, \infty) \backslash \{ b\}$
\begin{align}
v_{b}^\prime (x)=&\beta Z^{(\alpha)} (x)
+\beta \delta \alpha \int_0^{x-b}\bW^{(\alpha)}(y)W^{(\alpha)}(x-y) dy\notag\\
&-\frac{\beta Z^{(\alpha)} (b) -1+\beta \alpha \int_0^\infty e^{-\varphi (\alpha) y}W^{(\alpha)}(y+b)dy}
{ \varphi (\alpha)\int_0^\infty e^{-\varphi (\alpha) y}W^{(\alpha)}(y+b)dy}
\bigg{(} W^{(\alpha)} (x)\notag\\
&+ \delta \int_0^{x-b} \bW^{(\alpha)} (y) W^{(\alpha)\prime}(x-y)dy\bigg{)}\notag\\
&+r\int_{(0, b)}w'_+ (z) \Bigg{(} \frac{\int_b^\infty e^{-\varphi (\alpha) u}W^{(\alpha)}(u-z)du}{  \int_b^\infty e^{-\varphi (\alpha)y}W^{(\alpha)}(y)dy}
\rbra{W^{(\alpha)}(x)+ \delta \int_0^{x-b}\bW^{(\alpha)}(y)W^{(\alpha)\prime}(x-y)dy }\notag\\
&-\rbra{W^{(\alpha)}(x-z) +\delta\int_0^{x-b}\bW^{(\alpha)}(y)W^{(\alpha)\prime}(x-z-y)dy}
\Bigg{)}dz\notag\\
&+r\int_{(b, \infty )}w'_+ (z) \Bigg{(} \frac{e^{-\varphi (\alpha)z}}{ \varphi (\alpha)\delta \int_b^\infty e^{-\varphi (\alpha)y}W^{(\alpha)}(y)dy} \notag\\
&\times
\rbra{W^{(\alpha)}(x)+\delta \int_0^{x-b}\bW^{(\alpha)}(y)W^{(\alpha)\prime}(x-y)dy }
-\bW^{(\alpha) }(x-z)
\Bigg{)}dz\notag\\
&+\delta \bW^{(\alpha)}(x-b)g(b),\label{161}
\end{align}
where
\begin{align}
g(b):=&
\beta  Z^{(\alpha)}(b)-1
-\frac{\beta Z^{(\alpha)} (b) -1+\beta \alpha \int_0^\infty e^{-\varphi (\alpha) y}W^{(\alpha)}(y+b)dy}
{ \varphi (\alpha)\int_0^\infty e^{-\varphi (\alpha) y}W^{(\alpha)}(y+b)dy}
W^{(\alpha)}(b) \notag\\
&+r\int_{(0, b)}w'_+(z) \rbra{ \frac{\int_b^\infty e^{-\varphi (\alpha) u}W^{(\alpha)}(u-z)du}{  \int_b^\infty e^{-\varphi (\alpha)y}W^{(\alpha)}(y)dy} W^{(\alpha)}(b) - W^{(\alpha)}(b-z)
}dz\notag\\
&+r\int_{(b, \infty )}w'_+(z) \frac{e^{-\varphi (\alpha)z}}{ \varphi (\alpha)\delta \int_b^\infty e^{-\varphi (\alpha)y}W^{(\alpha)}(y)dy} W^{(\alpha)}(b) dz.\label{174}
\end{align}
Hence, it follows that $v_{b}^\prime (b+)- v_{b}^\prime(b-)=\delta \bW^{(\alpha)}(0)g(b)$. For the case of bounded variation, by Remark \ref{remark_scale_function_properties} (ii), we have that $\mathbb{W}^{(\alpha)}(0)>0$. Therefore we have that $v_b$ is continuously differentiable at $b$ if and only if $g(b)=0$.

For the case that the process $X$ has paths of unbounded variation, using the fact that $\mathbb{W}^{\alpha}(0)=0$ and dominated convergence, we obtain, for $x \in (0, \infty) \backslash \{ b\}$, that
\begin{align}
v_{b}^{\prime\prime} (x)=&\alpha\beta W^{(\alpha)} (x)
+\beta \delta \alpha \int_b^{x}\bW^{(\alpha)\prime}(x-y)W^{(\alpha)}(y) dy
\notag\\
&-\frac{\beta Z^{(\alpha)} (b) -1+\beta \alpha \int_0^\infty e^{-\varphi (\alpha) y}W^{(\alpha)}(y+b)dy}
{ \varphi (\alpha)\int_0^\infty e^{-\varphi (\alpha) y}W^{(\alpha)}(y+b)dy}\notag\\
&\times\rbra{ W^{(\alpha)\prime} (x)
	+ \delta \int_b^{x} \bW^{(\alpha)\prime} (x-y) W^{(\alpha)\prime}(y)dy}\notag\\
&+r\int_{(0, b)}w'_+ (z) \bigg{(} \frac{\int_b^\infty e^{-\varphi (\alpha) u}W^{(\alpha)}(u-z)du}{  \int_b^\infty e^{-\varphi (\alpha)y}W^{(\alpha)}(y)dy}
\rbra{W^{(\alpha)\prime}(x)+ \delta \int_b^{x}\bW^{(\alpha)\prime}(x-y)W^{(\alpha)\prime}(y)dy } \notag\\
&-\rbra{W^{(\alpha)\prime}(x-z) +\delta\int_b^{x}\bW^{(\alpha)\prime}(x-y)W^{(\alpha)\prime}(y-z)dy}
\bigg{)}dz \notag\\
&+r\int_{(b, \infty )}w'_+(z) \bigg{(} \frac{e^{-\varphi (\alpha)z}}{ \varphi (\alpha)\delta \int_b^\infty e^{-\varphi (\alpha)y}W^{(\alpha)}(y)dy} \notag\\
&\times
\rbra{W^{(\alpha)\prime}(x)+\delta \int_b^{x}\bW^{(\alpha)\prime}(x-y)W^{(\alpha)\prime}(y)dy }
-\bW^{(\alpha)\prime }(x-z)
\bigg{)}dz +\delta \bW^{(\alpha)\prime}(x-b)g(b). \label{sec_der_v}
\end{align}
This implies that $v''_b(b+)-v''_b(b-)=\delta \bW^{(\alpha)\prime}(0+)g(b)$, and hence $v_b$ is twice continuously differentiable at $b$ if and only if $g(b)=0$.
The above discussion, together with the smoothness of the scale function on $\R\backslash\{0\}$ as given in Remark \red{\ref{remark_scale_function_properties}} (i), implies the following result.
\begin{lemma}\label{thm:smoothfit}
Assume that there exists $b\geq0$ such that $g(b)=0$. Then, $v_b$ is continuously differentiable on $\R\backslash\{0\}$ if $X$ has paths of bounded variation. Similarly if $X$ has paths of unbounded variation  $v_b$ is twice continuously differentiable on $\R\backslash\{0\}$.
\end{lemma}
We continue to discuss the behavior of the function $v_b$ at zero.
\begin{remark}[Continuity/smoothness at zero]\ (i) Using Proposition \ref{enp_aux} 
	we have that $v_b$ is continuous at zero for any $b\geq0$.\\
(ii) Suppose that $b>0$ and $X$ is of unbounded variation, then \eqref{161} gives
\begin{align}
v_{b}^\prime (0+)=&\beta  -\Bigg[\frac{\beta Z^{(\alpha)} (b) -1+\beta \alpha \int_0^\infty e^{-\varphi (\alpha) y}W^{(\alpha)}(y+b)dy}
{ \varphi (\alpha)\int_0^\infty e^{-\varphi (\alpha) y}W^{(\alpha)}(y+b)dy}\notag\\
&-r\int_{(0, b)}w'_+ (z)  \frac{\int_b^\infty e^{-\varphi (\alpha) u}W^{(\alpha)}(u-z)du}{  \int_b^\infty e^{-\varphi (\alpha)y}W^{(\alpha)}(y)dy} \diff z\notag\\
&-r\int_{(b, \infty )}w'_+ (z)  \frac{e^{-\varphi (\alpha)z}}{ \varphi (\alpha)\delta \int_b^\infty e^{-\varphi (\alpha)y}W^{(\alpha)}(y)dy}\diff z\Bigg]W^{(\alpha)} (0+)=\beta=v_b'(0-).\label{v_prime_0}
\end{align}
\end{remark}
\subsection{Existence of the optimal threshold} \label{Sec105}
Let us define our candidate optimal threshold by
\begin{align}
b^\ast := \inf \{b\geq 0 : g(b) < 0 \} \label{1}
\end{align}
with the convention that $\inf\varnothing=+\infty$, where the function $g(\cdot)$ is defined in \eqref{174}.

Consider the refracted L\'evy process $\Gamma^{(b)}=(\Gamma^{(b)}(t): t\geq0)$ at the level $b\geq0$, given as the strong solution to the following SDE
\[
\Gamma^{(b)}(t)=X(t)-\delta\int_0^t1_{\{\Gamma^{(b)}(s)>b\}}\diff s,\qquad\text{$t\geq0$,}
\]
and
$
\kappa^-_0 :=\inf\{t\geq 0 : \Gamma^{(b)}(t) < 0\}.
$
\begin{lemma}
We have that $0\leq b^*<\infty$. Furthermore we have that $b^*=0$ if and only if $X$ has bounded variation paths and
\begin{align}\label{cond_b_0}
\beta-1\leq \rbra{\delta (\beta-1) +\frac{\beta \alpha}{\varphi (\alpha)}-r\int_{(0, \infty)}w'_+(z) e^{-\varphi (\alpha)z}dz}\frac{1}{c}.
\end{align}
\end{lemma}
\begin{proof}

First we note that integration by parts gives
\begin{align*}
	&\int_{(0, b)}w'_+ (z)  \frac{\int_b^\infty e^{-\varphi (\alpha) u}W^{(\alpha)}(u-z)du}{  \int_b^\infty e^{-\varphi (\alpha)y}W^{(\alpha)}(y)dy} W^{(\alpha)}(b)dz\\
	=&\int_{(0, b)}w'_+(z)  \frac{e^{-\varphi(\alpha)b}W^{(\alpha)}(b-z)+\int_b^\infty e^{-\varphi (\alpha) u}W^{(\alpha)\prime}(u-z)du}{ \varphi(\alpha) \int_b^\infty e^{-\varphi (\alpha)y}W^{(\alpha)}(y)dy} W^{(\alpha)}(b)dz\\
	=&\frac{\int_0^\infty e^{-\varphi (\alpha) y}W^{(\alpha)\prime}(y+b)dy}{\varphi(\alpha)\int_0^\infty e^{-\varphi (\alpha) y}W^{(\alpha)}(y+b)dy}\int_{(0, b)}w'_+(z)  \frac{\int_b^\infty e^{-\varphi (\alpha) u}W^{(\alpha)\prime}(u-z)du}{ \int_b^\infty e^{-\varphi (\alpha)y}W^{(\alpha)\prime}(y)dy} W^{(\alpha)}(b)dz\\
	&+\frac{e^{-\varphi(\alpha)b}W^{(\alpha)}(b)}{ \varphi(\alpha) \int_b^\infty e^{-\varphi (\alpha)y}W^{(\alpha)}(y)dy}\int_{(0, b)}w'_+ (z)   W^{(\alpha)}(b-z)dz,
\end{align*}
and
\begin{align*}
	&\int_{(b, \infty )}w'_+ (z) \frac{e^{-\varphi (\alpha)z}}{ \varphi (\alpha)\delta \int_b^\infty e^{-\varphi (\alpha)y}W^{(\alpha)}(y)dy} W^{(\alpha)}(b) dz\\
	=&\frac{\int_0^\infty e^{-\varphi (\alpha) y}W^{(\alpha)\prime}(y+b)dy}{\varphi(\alpha)\int_0^\infty e^{-\varphi (\alpha) y}W^{(\alpha)}(y+b)dy}\int_{(b, \infty )}w'_+ (z) \frac{e^{-\varphi (\alpha)z}}{ \delta \int_b^\infty e^{-\varphi (\alpha)y}W^{(\alpha)\prime}(y)dy} W^{(\alpha)}(b) dz.
\end{align*}
Putting all pieces together and using Theorem 6 (ii) in \cite{KL}, we obtain that
\begin{align}
	&\int_{(0, b)}w'_+ (z) \rbra{ \frac{\int_b^\infty e^{-\varphi (\alpha) u}W^{(\alpha)}(u-z)du}{  \int_b^\infty e^{-\varphi (\alpha)y}W^{(\alpha)}(y)dy} W^{(\alpha)}(b) - W^{(\alpha)}(b-z)
	}dz\notag\\
	&+\int_{(b, \infty )}w'_+(z) \frac{e^{-\varphi (\alpha)z}}{ \varphi (\alpha)\delta \int_b^\infty e^{-\varphi (\alpha)y}W^{(\alpha)}(y)dy} W^{(\alpha)}(b) dz\notag\\
	=&\frac{\int_0^\infty e^{-\varphi (\alpha) y}W^{(\alpha)\prime}(y+b)dy}{\varphi(\alpha)\int_0^\infty e^{-\varphi (\alpha) y}W^{(\alpha)}(y+b)dy}\mathbb{E}_b\left[\int_0^{\kappa_0^-}e^{-\alpha t}w'_+(\Gamma^{(b)}(t))dt\right]\notag\\
	&+\int_{(0, b)}w'_+(z)   W^{(\alpha)}(b-z)dz\Bigg(\frac{W^{(\alpha)}(b)}{ \varphi(\alpha) \int_0^\infty e^{-\varphi (\alpha)y}W^{(\alpha)}(y+b)dy}\notag\\
	&+\frac{\int_0^\infty e^{-\varphi (\alpha) y}W^{(\alpha)\prime}(y+b)dy}{\varphi(\alpha)\int_0^\infty e^{-\varphi (\alpha) y}W^{(\alpha)}(y+b)dy}-1\Bigg)\notag\\
	=&\frac{\int_0^\infty e^{-\varphi (\alpha) y}W^{(\alpha)\prime}(y+b)dy}{\varphi(\alpha)\int_0^\infty e^{-\varphi (\alpha) y}W^{(\alpha)}(y+b)dy}\mathbb{E}_b\left[\int_0^{\kappa_0^-}e^{-\alpha t}w'_+(\Gamma^{(b)}(t))dt\right],\label{resol_aux}
\end{align}
where the last equality follows by applying integration by parts.

By using \eqref{resol_aux} in \eqref{174}, we can write
\begin{align}\label{fun_g}
	g(b)
	=&\left(\beta  Z^{(\alpha)}(b)-1\right)\left(1-\frac{W^{(\alpha)}(b)}{\varphi (\alpha)\int_0^\infty e^{-\varphi (\alpha) y}W^{(\alpha)}(y+b)dy}\right)\notag\\
	&+r\frac{\int_0^\infty e^{-\varphi (\alpha) y}W^{(\alpha)\prime}(y+b)dy}{\varphi(\alpha)\int_0^\infty e^{-\varphi (\alpha) y}W^{(\alpha)}(y+b)dy}\mathbb{E}_b\left[\int_0^{\kappa_0^-}e^{-\alpha t}w'_+(\Gamma^{(b)}(t))dt\right]-\frac{\beta\alpha W^{(\alpha)}(b)}{\varphi(\alpha)}.
\end{align}
On the other hand, we note that integration by parts gives
\begin{align}\label{fun_h}
h(b):=\left(1-\frac{W^{(\alpha)}(b)}{\varphi (\alpha)\int_0^\infty e^{-\varphi (\alpha) y}W^{(\alpha)}(y+b)dy}\right)=\frac{\int_0^\infty e^{-\varphi (\alpha) y}W^{(\alpha)\prime}(y+b)dy}{\varphi(\alpha)\int_0^\infty e^{-\varphi (\alpha) y}W^{(\alpha)}(y+b)dy}>0.
\end{align}
In view of Theorem 5 (ii) in \cite{KL}, we obtain
\begin{align}\label{smooth_fit}
\frac{g(b)}{h(b)}&=\beta\mathbb{E}_b\left[e^{-{{\alpha }}\kappa_0^-};\kappa_0^-<\infty\right]-1+r\mathbb{E}_b\left[\int_0^{\kappa_0^-}e^{-\alpha t}w'_+(\Gamma^{(b)}(t))dt\right]\notag\\
&=\beta-1 - \bE_b \left[ \int_0^{\kappa^-_0} e^{-\alpha t} \rbra{\beta \alpha -rw'_+ (\Gamma^{(b)} (t))}dt \right].
\end{align}
Therefore, we can proceed verbatim as in the proof of Lemma 4.4 in \cite{HPY} to conclude that $b\mapsto g(b)/h(b)$ is non-increasing by the concavity of $w$.

By Exercise 8.5 (i) in \cite{K} and identity (3.8) in \cite{KPP}, we get
\begin{align}
\lim_{b\uparrow\infty}h(b)=\lim_{b\uparrow\infty}\frac{\int_0^\infty e^{-\varphi (\alpha) y}W^{(\alpha)\prime}(y+b)/W^{(\alpha)}(b)dy}{\varphi(\alpha)\int_0^\infty e^{-\varphi (\alpha) y}W^{(\alpha)}(y+b)/W^{(\alpha)}(b)dy}=\frac{\Phi(\alpha)}{\varphi(\alpha)}.\label{fun_h_lim_b}
\end{align}
Therefore, \eqref{smooth_fit} and \eqref{fun_h_lim_b} together with Assumption \ref{AA3} imply
\begin{align}
\lim_{b\uparrow \infty}g(b)=\lim_{b\uparrow \infty}\frac{g(b)}{h(b)} h(b)
=\rbra{-1+\frac{r}{\alpha}w'_+(\infty)}\frac{\Phi(\alpha)}{\varphi(\alpha)}<0. \label{finite_b}
\end{align}
Using the previous identity we conclude that $b^*<\infty$.

By \eqref{174}, we also have that
	\begin{align}
	g(0)=\beta-1 -\rbra{ \delta (\beta - 1)+\frac{\beta \alpha}{\varphi (\alpha)}
		-r\int_{(0, \infty)} w'_+ (z) e^{-\varphi (\alpha) z }dz }W^{(\alpha)} (0). \label{2}
	\end{align}
For the case $X$ is of unbounded variation, Remark \ref{remark_scale_function_properties} (ii) implies that $g(0)=\beta-1>0$ and hence $b^*>0$. For the case of bounded variation, by Remark \ref{remark_scale_function_properties} (ii), we know that $b^*=0$ if and only if \eqref{cond_b_0} holds, which completes the proof.
\end{proof}

\subsection{Verification of optimality}\label{Sec106}
In this section, we will prove the optimality of the refracted-reflected strategy at the refraction threshold $b^\ast$ defined in \eqref{1}, and obtain the value function of the stochastic control problem given in \eqref{vf_def}.

We will begin by providing a result which allows us to express the expected NPV associated to the refracted-reflected strategy at the candidate threshold $b^*$ in a more convenient form.
\begin{lemma}\label{lem4-1}
Let $b^*\geq 0$ be the threshold defined in \eqref{1}, for $x\geq0$, we have
\begin{align}
v_{b^\ast}^\prime(x)
=\beta - \bE_x\left[ \int_0^{\kappa^-_0} e^{-\alpha t} \rbra{\beta \alpha -rw'_+ (\Gamma^{(b^\ast)} (t))}dt \right]\label{v_prime_opt}.
\end{align}
\end{lemma}
\begin{proof}
(i) We first assume that $b^\ast > 0$. Using that $g(b^*)=0$ in \eqref{fun_g}, we obtain
	\begin{align}
	&\frac{\beta Z^{(\alpha)} ({b^\ast}) -1+\beta \alpha \int_0^\infty e^{-\varphi (\alpha) y}W^{(\alpha)}(y+{b^\ast})dy}
	{ \varphi (\alpha)\int_0^\infty e^{-\varphi (\alpha) y}W^{(\alpha)}(y+{b^\ast})dy}\notag\\
	=&\frac{\beta\alpha}{\varphi(\alpha)h(b^*)}+r\frac{(h(b^*)-1)}{W^{(\alpha)}(b^*)}\bE_{b^{\ast}}\left[\int_0^{\kappa^-_0} e^{-\alpha t}w'_+ (\Gamma^{({b^\ast})}(t))dt\right].\label{aux_g_1}
	\end{align}
By Theorem 5 (ii) in \cite{KL} and \eqref{fun_h}, we get that
\begin{align}
\beta\E_x\left[e^{-\alpha\kappa_0^-};\kappa_0^-<\infty \right]=&\beta Z^{(\alpha)} (x)
+\beta \delta \alpha \int_0^{x-{b^\ast}}\bW^{(\alpha)}(y)W^{(\alpha)}(x-y) dy\notag\\
&-\frac{\beta\alpha}{\varphi(\alpha)h(b^*)}
\rbra{ W^{(\alpha)} (x)
	+ \delta \int_0^{x-{b^\ast}} \bW^{(\alpha)} (y) W^{(\alpha)\prime}(x-y)dy}.\label{dc_aux}
\end{align}
Similarly, using Theorem 6 (ii) in \cite{KL} together with \eqref{fun_h}, we have
\begin{align}
&r\bE_x\left[\int_0^{\kappa^-_0} e^{-\alpha  t}w'_+(\Gamma^{({b^\ast})}(t))dt\right]\notag\\
=&-r\frac{(h(b^*)-1)}{W^{(\alpha)}(b^*)}\bE_{b^{\ast}}\left[\int_0^{\kappa^-_0} e^{-\alpha t}w'_+ (\Gamma^{({b^\ast})}(t))dt\right]\rbra{ W^{(\alpha)} (x)
	+ \delta \int_0^{x-{b^\ast}} \bW^{(\alpha)} (y) W^{(\alpha)\prime}(x-y)dy}\notag\\
&+r\int_{(0, b^{\ast})}w'_+(z) \bigg{[} \frac{e^{-\varphi(\alpha)b^*}W^{(\alpha)}(b^*-z)+\int_{b^{\ast}}^\infty e^{-\varphi (\alpha) u}W^{(\alpha)\prime}(u-z)du}{  \varphi(\alpha)\int_{b^{\ast}}^\infty e^{-\varphi (\alpha)y}W^{(\alpha)}(y)dy}\notag\\
&\times\rbra{W^{(\alpha)}(x)+ \delta \int_0^{x-b^{\ast}}\bW^{(\alpha)}(y)W^{(\alpha)\prime}(x-y)dy }\notag\\
&-\rbra{W^{(\alpha)}(x-z) +\delta\int_0^{x-b^{\ast}}\bW^{(\alpha)}(y)W^{(\alpha)\prime}(x-z-y)dy}
\bigg{]}dz\notag\\
&+r\int_{(b^{\ast}, \infty )}w'_+ (z) \bigg{(} \frac{e^{-\varphi (\alpha)z}}{ \varphi (\alpha)\delta \int_{b^{\ast}}^\infty e^{-\varphi (\alpha)y}W^{(\alpha)}(y)dy} \notag\\
&\times
\rbra{W^{(\alpha)}(x)+\delta \int_0^{x-b^{\ast}}\bW^{(\alpha)}(y)W^{(\alpha)\prime}(x-y)dy }
-\bW^{(\alpha) }(x-z)
\bigg{)}dz.\label{resol_ref_aux}
\end{align}
We now note that integration by parts gives
\begin{align}
\varphi(\alpha)\int_{b^{\ast}}^\infty e^{-\varphi (\alpha) u}W^{(\alpha)}(u-z)du=e^{-\varphi(\alpha)b^*}W^{(\alpha)}(b^*-z)+\int_{b^{\ast}}^\infty e^{-\varphi (\alpha) u}W^{(\alpha)\prime}(u-z)du.\label{int_der_W}
\end{align}
Therefore using \eqref{aux_g_1}, \eqref{dc_aux}, \eqref{resol_ref_aux} and \eqref{int_der_W} in \eqref{161}, we obtain \eqref{v_prime_opt}.

(ii) For the case $b^\ast = 0$, using \eqref{161} together with \eqref{111} and \eqref{111_der}, we obtain that for $x\geq0$,
\begin{align*}
v_{0}^\prime(x)
= \beta\rbra{\bZ^{(\alpha)}(x) -\frac{\alpha}{\varphi (\alpha)} \bW^{(\alpha)}(x)}
+r\int_{(0, \infty )}w'_+ (z) \rbra{ e^{-\varphi (\alpha)z} \bW^{(\alpha)}(x)
	-\bW^{(\alpha) }(x-z)}dz.
\end{align*}
Therefore, Theorems 5 (ii) and 6 (ii) in \cite{KL} together with \eqref{int_der_W} imply \eqref{v_prime_opt}.
\end{proof}
Let $\mathcal{L}$ be the infinitesimal generator associated with
the process $X$ applied to a $C^1$ (resp.\ $C^2$) function $f$ for the case $X$ is of bounded (resp.\ unbounded) variation:
\begin{align} \label{generator}
\mathcal{L} f(x) &:= \gamma f'(x) + \frac 1 2 \sigma^2 f''(x) + \int_{(-\infty,0)} \left[ f(x+z) - f(x) -  f'(x) z 1_{\{-1 < z < 0\}} \right] \Pi(\diff z), \quad x  >0.
\end{align}
We now provide a verification lemma.
The proof is essentially the same as Lemma 5.1 in \cite{PerYamYu2018} (which deals with the case where the payoff function $w$ is equal to zero) and Lemma 4.1 in \cite{HPY} (which deals with the case without capital injection), and is hence omitted.
\begin{lemma}[Verification lemma]\label{ver_lemma}
Suppose $\hat{\pi}$ is an admissible dividend strategy such that $v_{\hat{\pi}}$ is sufficiently smooth on $(0, \infty)$, continuous on $\R$, and, for the case of unbounded variation, continuously differentiable at zero. In addition, we assume that
\begin{align}
\sup_{0\leq r\leq \delta}((\cL -\theta)v_{\hat{\pi}}(x)-rv_{\hat{\pi}}'(x)+r)+w(x) \leq 0 ,\quad &~~~x>0, \label{189}\\
v_{\hat{\pi}}^\prime(x) \leq \beta,\quad&~~~x>0, \notag\\
\inf_{x \geq 0} v_{\hat{\pi}} (x) >-m,\quad&~~~\text{ for some }m>0.\notag
\end{align}
Then $v(x)=v_{\hat{\pi}}(x)$ for all $x\geq 0$ and hence $\hat{\pi}$ is an optimal strategy.
\end{lemma}

We shall proceed by obtaining some properties of the function $v_{b^*}$.
\begin{lemma} \label{lemma_slope}
	For the optimal threshold $b^*$ defined in \eqref{1}, we have $1\leq v_{b^*}'(x) \leq \beta$ for $x < b^*$, and $0 \leq v_{b^*}'(x) \leq 1$ for $x \geq b^*$.
\end{lemma}
\begin{proof}
Because $w$ is concave and its right derivative $w'_+(0+)\leq \beta$, we have that the mapping $x \mapsto \beta\alpha -rw'_+(x)$ is non-decreasing on $[0, \infty)$. It follows that the mapping
\[
x\mapsto \bE_x\left[ \int_0^{\kappa^-_0} e^{-\alpha t} \rbra{\beta \alpha -rw'_+(\Gamma^{(b^\ast)} (t))}dt \right],
\]
is non-decreasing as well. Using \eqref{v_prime_opt}, we derive that
\begin{align*}
v_{b^\ast}^\prime(x)&=\beta - \bE_x\left[ \int_0^{\kappa^-_0} e^{-\alpha t} \rbra{\beta \alpha -rw'_+(\Gamma^{(b^\ast)} (t))}dt \right]\\
&=\beta\E_x\left[e^{-\alpha\kappa_0^-};\kappa_0^-<\infty\right]+r\bE_x\left[ \int_0^{\kappa^-_0} e^{-\alpha t} w'_+(\Gamma^{(b^\ast)} (t))dt \right].
\end{align*}
By Assumption \ref{AA3} we have that the right derivative $w'_+(x)\geq0$ for all $x\geq0$, therefore the previous equality implies that the function $v_{b^\ast}^\prime$ is non-increasing and non-negative on $(0, \infty)$.

(i) Let us assume $b^*>0$. Using \eqref{smooth_fit}, and the fact that $g(b^*)=0$, we obtain
\begin{equation}\label{v_prime_b}
v_{b^\ast}^\prime(b^*)=\beta - \bE_{b^*}\left[ \int_0^{\kappa^-_0} e^{-\alpha t} \rbra{\beta \alpha -rw'_+ (\Gamma^{(b^\ast)} (t))}dt \right]=1.
\end{equation}
On the other hand, the fact that $w'_+(x)\leq \beta$ for all $x\geq0$, implies 
\begin{align}\label{v_prime_00}
v_{b^\ast}^\prime(0+)=\beta - \bE_{0}\left[ \int_0^{\kappa^-_0} e^{-\alpha t} \rbra{\beta \alpha -rw'_+ (\Gamma^{(b^\ast)} (t))}dt \right]\leq \beta - \beta\rbra{ \alpha -r}\bE_{0}\left[ \int_0^{\kappa^-_0} e^{-\alpha t} dt \right]\leq \beta.
\end{align}
Therefore using that $v_{b^\ast}^\prime$ is non-increasing on $(0, \infty)$ together with \eqref{v_prime_b} and \eqref{v_prime_00}, we obtain the result.

(ii) For the case $b^*=0$, we note that \eqref{smooth_fit} and the fact that $g(0)\leq 0$ yield that
\begin{align*}
v_{0}^\prime(0+)=\beta - \bE_0\left[ \int_0^{\kappa^-_0} e^{-\alpha t} \rbra{\beta \alpha -rw'_+(\Gamma^{(0)} (t))}dt \right]=\frac{g(0)}{h(0)}+1\leq 1 .
\end{align*}
This observation, together with the fact that $v_0'$ is non-increasing and non-negative on $[0, \infty)$, implies the result.
\end{proof}
\begin{remark}\label{inf_v_opt}
By Lemma \ref{lemma_slope}, we have that the function $v_{b^*}'$ is non-decreasing on $(0,\infty)$. This implies that $\inf_{x\geq 0}v_{b^*}(x)\geq v_{b^*}(0)>-\infty$.
\end{remark}
The next result allows us to work the HJB equation in a simplified way. The proof is essentially the same as for Lemma 4.3 in \cite{HPY} (without capital injection) and therefore is omitted.
\begin{lemma}\label{optimality_1}
	The first inequality in \eqref{189} holds for $v_{b^*}$ if and only if
	\begin{equation}\label{ineq_opt}
		\begin{cases}
			v_{b^*}'(x) \geq 1, & \text{if $0< x\leq b^*$}, \\
			v_{b^*}'(x) \leq 1, & \text{if $x>b^*$}.
		\end{cases}
	\end{equation}
\end{lemma}

\begin{proof}[Proof of Theorem \ref{Prop102}]
Thanks to Remark \ref{inf_v_opt}, Lemmas \ref{ver_lemma}, \ref{lemma_slope} and \ref{ineq_opt}, we verified the optimality of the refracted-reflected strategy at the threshold $b^*$, explicitly given in \eqref{1}, as in the statement.
\end{proof}
\section{Optimal strategies for bail-out dividend problems with regime switching}\label{sec5}
This section provides the rigorous proof of the Theorem \ref{Thm201} formulated in Section \ref{sec3}, which is the main result of this paper.

\subsection{Iteration algorithm to compute the value function $V$}
We will first show that the value function $V_{\pi^{b}}$, under a modulated refracted-reflected strategy at the level $b=(b(i),i\in E)$ and $0$ respectively, solves a fixed point equation.

Let us define
\begin{align*}
V_+:=\frac{\delta_+ }{r_-},\qquad\text{and}\qquad V_-:=-\beta \sup_{i\in E}\E_{(0,i)}\left[\int_0^{\infty}e^{-r_-t}\diff R^{(0)}(t)\right],
\end{align*}
where $\delta_+:= \max_{i\in E}\delta(i)$, 
$r_-:=\min_{i\in E}r(i)$ and $R^{(0)}(t) :=\sup_{s\leq t}(-(X(s)-\delta_+s))\vee 0$ for $t\geq0$.
It is straightforward to check the next auxiliary result, which provides some bounds for the value function that will be useful for later iterations.
\begin{proposition}\label{Prop201}
 We have that
\begin{align*}
V_-< {V(x, i)}< V_+,
\end{align*}
for all $x\in [0, \infty)$ and $i\in E$.
\end{proposition}

We then consider the space of functions
\begin{align*}
\cB :=&\{ f: f(\cdot,i)\in C
([0,\infty)), \text{$f(x,i)$ is bounded, for $i\in E$}\}.
\end{align*}
endowed with the norm $\|f\| :=\max_{i\in E}\sup_{x\geq 0}|f(x,i)|$ for $f\in \cB$.

For any $f:[0,\infty)\times E\mapsto\R$, we define the function $\tilde{f}:[0,\infty)\times E\mapsto\R$ given by
\begin{align}\label{finite_tilde}
\tilde{f}(x,i):=\sum_{j\in E,j\not=i}\frac{q_{ij}}{q_i}\int_{(-\infty,0)}\Big[(\beta(x+y)+f(0,j))1_{\{-y>x\}}+f(x+y, j)1_{\{-y\leq x\}}\Big]\diff  F_{ij}(y),
\end{align}
where $F_{ij}$ is the distribution function of the random variable $J_{ij}$ for $i,j\in E$ and $q_i = \sum_{j\neq i} q_{ij}$.

\begin{remark}\label{rem_tilde_f}
Note that
\begin{align*}
|\tilde{f}(x,i)|\leq \sum_{j\in E,j\not=i}\frac{q_{ij}}{q_i}\left[-\beta\E[J_{ij}]+\|f\|\right],\qquad\text{$(x,i)\in[0,\infty)\times E$}.
\end{align*}
Then if $f\in\cB$, Assumption \ref{AAA2} 
implies that $\tilde{f}\in\cB$.
\end{remark}
Now for any pair of functions $b\in\cE$ and $f\in\cB$, we define the following mapping
\begin{align}\label{opr_T_def}
T_b f(x, i):=
\bE^i_{x} \Bigg[\int_0^\infty e^{-\alpha_i t}\diff L_i^{0,b(i)}(t) -\beta\int_{[0, \infty)} e^{-\alpha_i t}dR_i^{0,b(i)}(t)+q_i\int_0^\infty e^{-\alpha_i t} \tilde{f}(U^{0,b(i)}_i(t), i)\diff t\Bigg]
\end{align}	
where $\alpha_i = r(i)+q_i$, and $\bE^i_x$ denotes the expectation operator associated to the law of the process $X^i$ conditioned on the event $\{X^i_0=x\}$.
The process $U^{0,b(i)}_i$
is the refracted-reflected process with threshold $b(i)\geq0$ driven by $X^i$, defined in Section \ref{Sec103}; and $L_i^{0,b(i)}(t)$, 
$R_i^{0,b(i)}(t)$
are its corresponding cumulative dividend payments and capital injection, respectively.
\begin{proposition}\label{Prop203}
For $b\in\cE$, and $(x,i)\in\R\times E$ we have
\begin{align*}
V_{\pi^{b}}(x,i) = T_b V_{\pi^{b}}(x,i).
\end{align*}
\end{proposition}
\begin{proof}
By Proposition \ref{Prop201}, together with Lemma \ref{Lipschitz_con_vf}, we have that $V_{\pi^{b}}(x,i)$ is bounded and continuous and hence $V_{\pi^{b}}(x,i)\in\cB$. Let $\zeta$ be the epoch of the first regime switch. By an application of the strong Markov property, we obtain
\begin{align}
 V_{\pi^{b}} (x,i) =& \mathbb{E}_{(x,i)} \left[ \int_0^\infty e^{-\int_0^tr(H(s))ds}\diff L^{0,b}(t) -  \int_{[0, \infty)} e^{-\int_0^tr(H(s))ds} \beta \diff R^{0,b}(t)\right]\notag\\
=&\bE_{(x, i)} \Bigg{[}\int_0^{\zeta } e^{- r(i) t}l^{0,b}(t) dt
-\beta \int_{[0, \zeta)} e^{-r(i) t}dR^{0, b}(t)\notag\\
&+e^{-r (i) \zeta }\bigg[\beta (U^{0,b}(\zeta-)+J_{iH(\zeta)})
+\bE_{(0, H({\zeta}))} \bigg{[}\int_0^{\infty } e^{- q \int_0^{t}r(H(s)) ds}l^{0,b}(t) dt\notag\\
&-\int_{[0, \infty)} e^{- \int_0^{t} qr(H(s)) ds}\beta \diff R^{0,b}(t) \bigg{]}\bigg]1_{\{U^{0,b}(\zeta-)< -J_{ij}\}}\notag\\
&+
e^{-r(i)\zeta}\bE_{(U^{0,b}(\zeta), H({\zeta}))} \bigg{[}\int_0^{\infty } e^{- q \int_0^{t}r(H(s)) ds}l^{0,b}(t) dt\notag\\
&-\int_{[0, \infty)} e^{- \int_0^{t} qr(H(s)) ds}\beta \diff R^{0,b}(t) \bigg{]}1_{\{U^{0,b}(\zeta-)\geq -J_{ij}\}}\Bigg{]}\notag\\
=&\bE_{(x, i)} \bigg{[}\int_0^{\zeta } e^{- r(i) t}l^{0,b}(t) dt
-\beta\int_{[0, \zeta)} e^{-r(i) t}dR^{0,b}(t)\notag\\
&+e^{-r(i)\zeta}V_{\pi^{0,b}}(U^{0,b}(\zeta), H(\zeta))1_{\{U^{0,b}(\zeta-)\geq -J_{ij}\}} \notag\\
&+e^{-r (i) \zeta }\Big(\beta (U^{0,b}(\zeta-)+J_{iH(\zeta)})+V_{\pi^{0,b}}(0, H(\zeta))\Big)1_{\{U^{0,b}(\zeta-)< -J_{ij}\}}\bigg{].}\label{fix-point_1}
\end{align}
Recall that the random variable $J_{ij}$ describes the jump when $H$ changes from the state $i$ to the state $j$. By conditioning on the state of the Markov chain $H$ at the first regime switching time $\zeta$ and the random variable $J_{ij}$, we can get
\begin{align}
&\bE_{(x, i)} \bigg[\int_0^{\zeta } e^{- r(i) t}l^{0,b}(t) dt
-\beta\int_{[0, \zeta)} e^{-r(i) t}dR^{0,b}(t)+e^{-r(i)\zeta}V_{\pi^{0,b}}(U^{0,b}(\zeta), H(\zeta))1_{\{U^{0,b}(\zeta-)\geq -J_{ij}\}} \notag\\
&+e^{-r (i) \zeta }\Big(\beta (U^{0,b}(\zeta-)+J_{iH(\zeta)})+V_{\pi^{0,b}}(0, H(\zeta))\Big)1_{\{U^{0,b}(\zeta-)< -J_{ij}\}}\bigg]\notag\\
=&\sum_{j\in E,j\not=i}\frac{q_{ij}}{q_i}\bE_{(x, i)} \bigg{[}\int_0^{\zeta } e^{- r(i) t}l^{0,b}(t) dt
-\beta\int_{[0, \zeta)} e^{-r(i) t}dR^{0,b}(t)\notag\\
&+e^{-r(i)\zeta}\Big[V_{\pi^b}(U^{0,b}(\zeta-)+J_{ij}, j)1_{\{U^{0,b}(\zeta-)\geq -J_{ij}\}}\notag\\
&+\Big(\beta(U^{0,b}(\zeta-)+J_{ij})+V_{\pi^b}(0,j)\Big)1_{\{U^{0,b}(\zeta-)<- J_{ij}\}}\Big]\bigg|H_{\zeta}=j\bigg{]}\notag\\
=&\bE_{(x, i)} \bigg{[}\int_0^{\zeta } e^{- r(i) t}l^{0,b}(t) dt
-\beta\int_{[0, \zeta)} e^{-r(i) t}dR^{0,b}(t)
+e^{-r(i)\zeta}\tilde{V}_{\pi^b}(U^{0,b}(\zeta-), i)\bigg{]}=T_bV_{\pi^b}(x, i).\label{fix-point_2}
\end{align}
The last equality follows by noting that $\zeta$ is an exponential random variable with rate $q_i$, independent of the processes $L^{0,b}$, $R^{0,b}$, and $U^{0,b}$.
\end{proof}
\begin{corollary}\label{Cor201}
The operator $T_b$ is a contraction on {$\cB$} with respect to the norm $\norm{\cdot}$. In particular,
for {$f \in \cB$}, we have
\begin{align}
V_{\pi^b}(x, i)= \lim_{n\uparrow \infty}T_b^n(f)(x, i), \qquad\text{$(x,i)\in [0, \infty)\times E$}, \label{8}
\end{align}
{where the convergence is in the $\|\cdot\|$-norm and $T_b^n(f) :=T_b(T_b^{n-1}(f))$ for $n>1$ with $T_b^1 :=T_b$.}
\end{corollary}
\begin{proof}
It is easy to check that $\cB$ endowed with the norm $\norm{ \cdot}$ is a complete metric space. Using \eqref{opr_T_def} together with Remark \ref{bounded_v} and Remark \ref{rem_tilde_f}, we have that $T_b$ is a mapping from $\cB$ to itself.
Hence for $f, g \in \cB$, we have
\begin{align}
\norm{T_b f- T_b g}=&\sup_{x \geq 0, i\in E}\bE_{(x,i)}\Bigg[ e^{- {r(i)}{\zeta}}  {\sum_{j\in E,j\not=i}} {\frac{q_{ij}}{q_i}} {\int_{(-\infty,-U^{0,b}(\zeta-))}}\absol{f(0, j)-g(0, j)   )}\diff  F_{ij}(y)\notag\\
&+e^{- {r(i)}{\zeta}}   {\sum_{j\in E,j\not=i}}{\frac{q_{ij}}{q_i}} {\int_{[-U^{0,b}(\zeta-),0]}}\absol{f(U^{0,b}(\zeta-)+y, j)-g(U^{0,b}({\zeta})+y, j)   )}\diff  F_{ij}(y)\Bigg]\notag\\
\leq& \norm{f-g}  \bE_{(0,i)}\Bigg[ e^{- {r(i)}{\zeta}}   {\sum_{j\in E,j\not=i}} {\frac{q_{ij}}{q_i}} {\int_{(-\infty,0]}}\diff  F_{ij}(y)\Bigg]\notag\\
\leq& \sup_{ i\in E}\bE_{(0,i)}\left[ e^{- {r(i)}{\zeta}}\right] \norm{f-g}{<\norm{f-g}}.\label{7}
\end{align}

By \eqref{7} for any $f\in\cB$, $(T_b^n f)_{n\geq 1}$ is a Cauchy sequence.
Therefore, by the continuity of the mapping $T_b:\cB\mapsto\cB$, we have
\[
T^{\infty}_bf:=\lim_{n\uparrow\infty}T^n_bf=T_b(\lim_nT^{n}_bf)=T_b(T^{\infty}_b(f)),\qquad\text{$f\in\cB$.}
\]
The previous identity implies that $T^{\infty}_bf$ is a fixed point for the mapping $T_b$, hence by Proposition \ref{Prop203}, we obtain \eqref{8}.
\end{proof}
\subsection{Verification by fixed point arguments}
Let us then define the following space of functions:
\begin{align*}
\cD:=\ \{ f \in \cB :&\text{$\tilde{f}(\cdot,i)$ is concave and satisfies that $\tilde{f}'_+ (0+, i) \leq \beta$ and $\tilde{f}'_+(\infty,i)\in[0,1]$  for $i\in E$}\}.
\end{align*}

The next result gives sufficient conditions for a function $f$ to belong to the class $\cD$.
\begin{proposition}\label{classD}
Consider $f\in\cB$ such that $f(\cdot,i)\in C^1((0,\infty))$, it is concave, non-decreasing, and satisfies that $f^\prime \leq \beta$ 
for $i\in E$. Then $f\in\cD$.
\end{proposition}
\begin{proof}
Using \eqref{finite_tilde} and integration by parts we have
\begin{align*}
\tilde{f}(x,i)=&\sum_{j\in E,j\not=i}\frac{q_{ij}}{q_i}\bigg(\int_{[-x,0)}\Big[f(x+y, j)-(\beta(x+y)+f(0,j))\Big]\diff  F_{ij}(y)\\
&+\beta(x+\E(J_{ij}))+f(0,j)\bigg)\\
=&\sum_{j\in E,j\not=i}\frac{q_{ij}}{q_i}\left(f(x, j)+\beta\E(J_{ij})-\int_{-x}^0\Big[f^\prime(x+y, j)-\beta)\Big] F_{ij}(y)dy\right),\ \text{$(x,i)\in[0,\infty)\times E$.}
\end{align*}
Using the condition that $f\in C^1((0,\infty))$, we obtain by differentiating the above expression 
\begin{align*}
\tilde{f}^\prime(x,i)={\sum_{j\in E,j\not=i}}{\frac{q_{ij}}{q_i }}\left[\beta+  {\int_{-x}^0\Big(f^{\prime}(x+y,j)-\beta\Big)F'_{ij}(y)\diff y}\right],\qquad\text{$i\in E$ and $x\geq0$.}
\end{align*}
Clearly, the fact that $f\in C^1((0,\infty))$ implies that $\tilde{f}\in C^1((0,\infty))$ as well. On the other hand, by the concavity of $f(\cdot,i)$ and the fact that $f^\prime(0+,j)\leq \beta$, we obtain that $\tilde{f}^\prime(x,i)$ is non-increasing and hence $\tilde{f}(\cdot,i)$ is concave as well.
Moreover, it is clear that $\tilde{f}^\prime(0+,i)=\sum_{j\in E,j\not=i}\frac{q_{ij}}{q_i}\beta\leq\beta$. On the other hand, by dominated convergence theorem,
\begin{align*}
\lim_{x\to\infty}\tilde{f}^\prime(x,i)=\sum_{j\in E,j\not=i}\frac{q_{ij}}{q_i}f^\prime(\infty,j) =0\leq 1,
\end{align*}
where we used the fact that $f^\prime(\infty,j)=0$ for all $j$ as $f(\cdot,j)$ is bounded, concave and nondecreasing. Therefore, by the previous arguments together with Remark \ref{rem_tilde_f}, we have that $f\in\cD$.
\end{proof}
For $f\in\cD$ and $(x,i)\in[0,\infty)\times E$, let us define another auxiliary operator:
\begin{align}
\Theta f(x, i)&:=\sup_{\pi \in \Pi} \bE_{(x, i)} \bigg{[}\int_0^{{\zeta} } e^{- {r(i)} t}l^\pi(t) dt
-\beta\int_{[0, {\zeta})} e^{-{r(i)} t}dR^\pi(t)
+ e^{-{r(i)}{\zeta}}\tilde{f}(U^\pi({\zeta-}), H({\zeta-}))\bigg{]}\notag\\
&=\sup_{\pi \in \Pi}\bE^i_x  \Bigg[\int_0^\infty e^{-\alpha_i t}\diff L_i^{\pi}(t) -\beta\int_{[0, \infty)} e^{-\alpha_i t}dR_i^{\pi}(t)+q_i\int_0^\infty e^{-\alpha_i t} \tilde{f}(U^{\pi}_i(t), i)dt\Bigg]\label{fun_U}.
\end{align}
\begin{remark}\label{rem_U_T}
Assume $f\in\cD$.
For fixed $i\in E$, Theorem \ref{Prop102} 
guarantees the existence of $b^f(i)\geq0$ such that the second equality of \eqref{fun_U} is achieved by the NPV of a refracted-reflected strategy at the barrier $b^f(i)$. Therefore, if we define $b^{f}=(b^f(i);i\in E)$, it follows that $\Theta f(x, i)=T_{b^f}f(x,i)$ for $(x,i)\in[0,\infty)\times E$.

The previous argument, together with Lemma \ref{thm:smoothfit} and the proof of Lemma \ref{lemma_slope}, implies that for $f\in\cD$ and $i\in E$, $\Theta f(\cdot, i)\in C^{1}((0,\infty))$, it is concave, $(\Theta f)'(0+, i)\leq \beta$, and $(\Theta f)'(\infty, i)\leq 1$. Hence, using Proposition \ref{classD}, we conclude that $\Theta f\in\cD$.
\end{remark}

\begin{proposition}\label{Prop301}
Let $v^-_0, v^+_0 \in \cD$ and we define $v^-_n\ :=\ \Theta v^-_{n-1}$ and $v^+_n\ :=\ \Theta v^+_{n-1}$ for $n\in\bN$, respectively. If $v^-_0 \leq V \leq v_0^+$, then we have $v^-_n\leq V \leq v^+_n$ for all $n\geq 1$, and
\begin{align}
V(x,i) =\lim_{n\uparrow \infty} v^-_n (x, i)=\lim_{n\uparrow \infty} v^+_n (x, i),~~~~~~
\text{ for $(x,i)\in[0,\infty)\times E$}, \label{10}
\end{align}
where the convergence is in the $\|\cdot\|$-norm. In particular $V\in\cD$.
\end{proposition}
\begin{proof}
(i) From the definition of $\Theta$ given in \eqref{fun_U}, we have that if $v_{n-1}^- \leq V \leq v_{n-1}^-$, then
\[v_n^-=\Theta v_{n-1}^- \leq \Theta V \leq \Theta v_{n-1}^+=v_n^+.\]
This, together with \eqref{4} implies that $v_n^- \leq V \leq v_n^+$. By induction, we obtain the first claim.
\ \\
(ii) By Remark \ref{rem_U_T}, we have that for any $f\in\cD$, there exists $b^f\in\cE$ such that $\Theta f=\sup_{b\in\cE}T_{b}f=T_{b^f}f$. Therefore, by the proof of Corollary \ref{Cor201}, it follows that for $f,g\in\cD$,
\begin{align*}
\norm{\Theta f-\Theta g}\leq \sup_{b \in \cE}\norm{T_b f -T_b g}\leq  \sup_{i \in E}\bE_{(0, i)}\left[e^{-q(i)\zeta} \right]\norm{f-g}.
\end{align*}
Hence the sequences $(v^+_{n})_{n\geq 1}$ and $(v^-_{n})_{n\geq 1}$ should converge to the unique fixed point of the mapping $\Theta$ which, by Proposition \ref{Prop101}, is given by $V$. Moreover, the existence of functions $v_0^+,v_0^-\in\cD$ is already given by Proposition \ref{Prop201}. In view of Remark \ref{rem_U_T}, it follows that functions $v_n^{\pm}$ belong to $\cD$.
On the other hand, using \eqref{finite_tilde} together with dominated convergence theorem, we note that for $(x,i)\in[0,\infty)\times E$
\begin{align}\label{lim_V}	
\tilde{V}(x,i)&=\lim_{n\to\infty}\sum_{j\in E,j\not=i}\frac{q_{ij}}{q_i}\int_{(-\infty,0)}\Big[(\beta(x+y)+v_n^{\pm}(0,j))1_{\{-y>x\}}+v_n^{\pm}(x+y, j)1_{\{-y\leq x\}}\Big]\diff  F_{ij}(y)\notag\\
&=\lim_{n\to\infty}\tilde{v}^{\pm}_n(x,i). 	
\end{align}
Therefore using the fact that the functions $v_n^{\pm}$ belong to $\cD$ together with \eqref{lim_V},
we derive that $V$ belongs to $\cD$ as well, which completes the proof. 
\end{proof}

The previous result gives an iterative construction of the value function as follows: Initializing by $n=0$ and $v=v_0$ for some $v_0\in\cD$, we can operate the iteration scheme by
\begin{itemize}
\item[(i)] Find $b^v=(b^v(i);i\in E)$ as in Remark \ref{rem_U_T};
\item[(ii)] Set $T_{b^v}v\to v$, $n+1\to n$, and $v\to v_n$ and return to step (i).
\end{itemize}

\subsection{Proof of Theorem \ref{Thm201}}
Proposition \ref{Prop301} gives that $V\in\cD$. 
Hence Proposition \ref{Prop101}, together with Remark \ref{rem_U_T}, implies that there exists $b^*\in\cE$ such that $V(x,i)=\Theta V(x,i)=T_{b^*}V(x,i)$ for $(x,i)\in[0,\infty)\times E$. 
This further yields that
\begin{align*}
V(x, i)=\lim_{n\uparrow \infty}T_{b^\ast}^n V(x, i), \qquad\text{for $(x,i)\in[0,\infty)\times E$.} 
\end{align*}
Finally, an application of Corollary \ref{Cor201} concludes that
\begin{align*}
V(x,i)
=V_{\pi^{b^*}}(x,i),\qquad\text{for $(x,i)\in[0,\infty)\times E$.}
\end{align*}
\qed

\appendix
\section{Proof of Auxiliary Results}\label{secA}
\subsection{Proof of technical lemmas}
In this section, we proceed to prove Proposition \ref{Prop101}. Before presenting the proof, we first need the preparation of the two technical lemmas.
\begin{lemma}\label{Lipschitz_con_vf}
For $y\geq x\geq 0$ and $i\in E$, we have
\begin{align}\label{ineq_q}
0\leq V(y,i)-V(x,i)\leq \beta (y-x).
\end{align}
In particular, it follows that, for fixed $i\in E$, $V(\cdot,i)$ is non-decreasing and Lipschitz-continuous.
\end{lemma}
\begin{proof}
	(i) First, we will prove that $V(\cdot,i)$ is non-decreasing. Let $\pi^{(\varepsilon)}=\{\ell^{\pi^{(\varepsilon)}},R^{\pi^{(\varepsilon)}}\}$ 
	be an $\varepsilon$-optimal strategy for $(U^{\pi^{(\varepsilon)}}(0),H(0))=(x,i)$. For $0\leq x<y$, we define the following strategy $\pi^y=\{\ell^{\pi^y},R^{\pi^y}\}$ given for $t\geq0$,
	\begin{align*}
	\ell^{\pi^y}(t):=\ell^{\pi^{(\varepsilon)}}(t)\qquad\text{and}\qquad R^{\pi^y}(t):=(R^{\pi^{(\varepsilon)}}(t)-y+x)\vee 0.
	\end{align*}
	The previous definition guarantees that $\pi^y$ is an admissible strategy for $(U^{\pi^y}(0),H(0))=(y,i)$. Note that
	\begin{align*}
	V(y,i)-V(x,i)&\geq V_{\pi^y}(y,i)-V_{\pi^{(\varepsilon)}}(x,i)-\varepsilon\\
	&=\mathbb{E}_{(x,i)} \left[  \int_{[0, \infty)} e^{-\int_0^tr(H(s))ds} \beta 1_{\{R^{\pi^{(\varepsilon)}}(t)\leq y-x\}}\diff R^{\pi^{(\varepsilon)}}(t)\right]-\varepsilon.
	\end{align*}
	Therefore, by taking $\varepsilon\to 0$, we deduce that
	\[
	V(y,i)- V(x,i)\geq0,\qquad\text{$0\leq x<y$.}
	\]
	(ii) Next, we prove the upper bound in \eqref{ineq_q}. Let $\pi^{(y,i)}=\{\ell^{(y,i)},R^{(y,i)}\}$ be an $\varepsilon$-optimal strategy for $(U^{\pi^{(y,i)}}(0),H(0))=(y,i)$. For $y\geq x\geq 0$, we define the strategy $\pi^{(x,y,i)}=\{\ell^{(x,y,i)},R^{(x,y,i)}\}$ given for $t\geq0$ by
	\begin{align*}
	\ell^{(x,y,i)}(t):=\ell^{(y,i)}(t)1_{\{t>0\}}\quad\text{and }\quad R^{(x,y,i)}(t):=(y-x)1_{\{t=0\}}+R^{(y,i)}(t)1_{\{t>0\}}.
	\end{align*}
	Hence
	\begin{align*}
	V(x,i)\geq V_{\pi^{(x,y,i)}}(x,i)\geq V_{\pi^{(y,i)}}(y,i)-\beta(y-x)\geq -\beta(y-x)+V(y,i)-\varepsilon.
	\end{align*}
	Therefore, by taking $\varepsilon\to 0$, we obtain that
	\[
	V(y,i)-V(x,i)\leq \beta(x-y),\qquad\text{$0\leq x\leq y$.}
	\]
\end{proof}
The proof of the next result follows closely that of \cite[Proposition 2.2]{JiaPis2012} with minor modifications required for the setting of absolutely continuous dividend strategies and capital injections.
\begin{lemma}\label{LemA01}
For all $\epsilon>0$ and $M>0$, there exists a strategy $\tilde{\pi}$ such that
\begin{align}\label{V_cont_mod}
\max_{i\in E} \sup_{x \in [0, M]}\rbra{  V (x, i) - V_{\tilde{\pi} } (x, i) }< \epsilon.
\end{align}
\end{lemma}
\begin{proof}
Let us take a partition of the interval $[0,M]$ as $\left(x_j:=\frac{jM}{N},j=0,\dots,N\right)$, where $N>M\varepsilon^{-1}$ that satisfies
\begin{align*}
\max_{i\in E}\sup_{|x-y|<M/N,x,y\in[0,M]}|V(x,i)-V(y,i)|<\varepsilon.
\end{align*}
Let $\pi^{i,j}=\{\ell^{\pi^{i,j}},R^{\pi^{i,j}}\}$ be $\varepsilon$-optimal strategies for $U^{\pi^{i,j}}(0)=x_j$ and $H(0)=i$, i.e., $V(x_j,i)-V_{\pi^{i,j}}(x_j,i)<\varepsilon$. For $x\in[0,M]$, we define the strategy $\tilde{\pi}=\{\ell^{\tilde{\pi}},R^{\tilde{\pi}}\}$ such that $U^{\tilde{\pi}}(0)=x$ and $H(0)=i$ given for $t\geq0$ by
\begin{align*}
\ell^{\tilde{\pi}}(t):=\ell^{\pi^{i,j^*}}(t),\qquad\text{and}\qquad R^{\tilde{\pi}}(t):=\beta(x_{j^*}-x)1_{\{t=0\}}+R^{\pi^{i,j^*}}(t)1_{\{t>0\}},
\end{align*}
where $j^*=\min\{j\geq 0:x\leq x_j\}$. Then, it follows that $|V_{\pi^{i,j^*}}(x_{j^\ast },i)-V_{\tilde{\pi}}(x,i)|\leq\beta(x_{j^*}-x)\leq 
\beta\varepsilon$.\\
Therefore we have
\begin{align*}
|V(x,i)-V_{\tilde{\pi}}(x,i)|&\leq |V(x,i)-V(x_{j^*},i)|+|V(x_{j^*},i)-V_{\pi^{i,j^*}}(x_{j^*},i)|+|V_{\pi^{i,j^*}}(x_{j^*},i)-V_{\tilde{\pi}}(x,i)|\\&=(2+\beta)\varepsilon.
\end{align*}
As the previous inequality is valid for arbitrary $x\geq0$ and $i\in E$, \eqref{V_cont_mod} is verified.
\end{proof}

\subsection{Proof of Proposition \ref{Prop101}}
Using the strong Markov property and proceeding as in the proof of \cite[Proposition 2.2]{JiaPis2012}, we obtain for $x\geq0$ and $i\in E$ that
\[
\sup_{\pi \in \Pi}
\bE_{(x, i)} \bigg{[}\int_0^{{\zeta} } e^{- {\Lambda(t) }}l^\pi(t) dt
-\int_{[0, {\zeta}]} {\beta}e^{-{\Lambda(t)}}dR^\pi(t)+
e^{-\Lambda (\zeta)}V(U^\pi (\zeta), H(\zeta))\bigg{]}\geq V(x,i).
\]
To prove the opposite inequality, let us fix the admissible strategy $\pi=\{\ell^{\pi},R^{\pi}\}$ and $\epsilon>0$.
By Lemma \ref{LemA01}, for all $k\in\bZ \cap [0, \infty)$, there exists $\pi_\epsilon^{ k}$ such that
\begin{align*}
\max_{i\in E}\sup_{x\in [k M, (k+1)M]} \rbra{V(x, i)-V_{\pi_\epsilon^{k}(x, i)}}<\epsilon.
\end{align*}
Let us denote $\theta$ as the shift operator. Recalling that $\zeta$ denotes the epoch of the first regime switch, we can define a strategy $\pi_\epsilon=\{\ell^{\pi_\epsilon},R^{\pi_\epsilon}\}$ as follows:
\begin{align*}
&l^{\pi_\epsilon} (t) :=  l^\pi (t)1_{\{ t <\zeta\}} + \sum_{k=0}^\infty
\rbra{l^{\pi_\epsilon^{k}} (t- \zeta) \circ \theta_{\zeta}  }1_{\{t \geq \zeta\}}1_{\{ U^\pi (\zeta) \in [k, (k+1))\}},\\
&R^{\pi_\epsilon} (t) :=  R^\pi (t)1_{\{ t <\zeta\}} + \sum_{k=0}^\infty
\rbra{R^{\pi_\epsilon^{k}} (t- \zeta) \circ \theta_{\zeta}  }1_{\{ t\geq \zeta\}}1_{\{ U^\pi (\zeta) \in [k, (k+1))\}}.
\end{align*}
Then we have
\begin{align*}
&\bE_{(x,i)} \left[ \int_0^{\zeta}e^{-r(i) t}l^{\pi}(t)dt-\beta \int_{[0, \zeta]}e^{-r(i)t}dR^{\pi}(t)
+e^{-r(i)\zeta}V(U^\pi (\zeta) , H(\zeta))\right]   \\
\leq&\bE_{(x,i)} \Bigg{[} \int_0^{\zeta}e^{-r(i) t}l^{\pi}(t)dt-\beta \int_{[0, \zeta]}e^{-r(i)t}dR^{\pi}(t)\notag\\
&+e^{-r(i)\zeta}\sum_{k=0}^\infty V_{\pi_\epsilon^{k}}(U^\pi (\zeta) , H(\zeta))
1_{\{ U^\pi (\zeta) \in [k, (k+1))\}} +\epsilon\Bigg{]} \\
=&V_{\pi_\epsilon} (x,i)+\epsilon \leq V(x,i) + \epsilon.
\end{align*}
As $\varepsilon>0$ is arbitrary, the proof is completed.
\qed

\end{document}